\documentclass[12pt]{report}

\usepackage{hyperref}
\usepackage{listings}
\usepackage[margin=3cm]{geometry}
\usepackage{amsmath}
\usepackage{amssymb}
\usepackage{mathtools}
\usepackage{amsthm}
\usepackage{csquotes}
\usepackage{subcaption}
\usepackage{framed}
\usepackage{graphicx}
\usepackage{graphviz}
\usepackage{algorithm}
\usepackage{algpseudocode}
\usepackage{tocloft}
\usepackage{authblk}
\usepackage{tikz}

\newcommand{\listalgolistingname}{List of Algorithms}
\newlistof{algolisting}{algs}{\listalgolistingname}

\newcommand{\algolisting}[2]{
\refstepcounter{algolisting}
\caption{\label{#1}#2}
\addcontentsline{algs}{algolisting}{\thealgolisting \hspace{0.2cm} #2}
}

\title{Specification of State and Time Constraints for Runtime Verification of Functions\\
\vspace{30px}\Large Technical Report}
\author[1, 2, $\dag$]{Joshua Dawes}
\author[1]{Giles Reger}

\affil[1]{University of Manchester, Manchester, UK}
\affil[2]{CERN, Geneva, Switzerland}
\affil[$\dag$]{\url{joshua.dawes@cern.ch}}

\newtheorem{definition}{Definition}[section]

\newtheorem{remark}{Remark}[section]
\newtheorem{example}{Example}[section]
\newtheorem{prop}{Proposition}[section]

\numberwithin{equation}{section}
\numberwithin{figure}{section}

\begin{document}

\maketitle
\tableofcontents
\newpage

\chapter{Motivation}\label{chapter-motivation}

This report presents the foundations for a body of work in Runtime Verification\cite{bartocci18}.  The work presented is therefore concerned with monitoring a program with respect to some specification.  The contributions include i) a new static model of programs that preserves reachability information, ii) a new logic, Control Flow Temporal Logic (CFTL), along with a CFTL-specific instrumentation technique and iii) an efficient monitoring algorithm.  CFTL is characterised by its tight coupling with the control flow of programs being verified, leading to a departure from the conventionally high level of abstraction in specification languages.

The material described in this report is the basis of work that will be performed at CERN, monitoring infrastructure of the Compact Muon Solenoid (CMS) Experiment, and will be described in a future paper.  CERN, the European Organisation for Nuclear Research, is a Particle Physics research laboratory in Geneva, Switzerland.  Part of CERN's accelerator complex is the Large Hadron Collider (LHC), a circular proton-proton collider whose design energy is 14 TeV.

During LHC runs, collisions/events result in unprecedented amounts of data, meaning understanding of the performance of any system that deals with the data must be precise.  With upgrades of the LHC planned that will increase luminosity (and, in turn, the amount of data to be dealt with), precise understanding becomes more important and so does a scalable method of obtaining it.

The first application will be a group of web services used by CMS.  These web services make a good initial test case for this work in Runtime Verification; they are object-oriented, involve local computation, communicate with other machines over a network and, most significantly, are subject to state and time constraints.  In addition, they are written in Python (a language common inside CMS) which has powerful introspection features of which advantage is taken in this work.

The web services of interest are for management and upload to databases of the so-called non-Event (Conditions) data, namely alignment and calibrations constants describing the CMS detector. Understanding errors and drops in performance of the Conditions upload service involves understanding performance at the source code level which, in the context of the test case this report describes, also leads to understanding of the performance of the surrounding network (since other machines must be queried to obtain certain information).  At the source code level, we must understand the relationships between data generated, the control flow that generated it and the timings involved.  In fact, changes in timing are often indicative of network behaviour in a service that communicates with other machines regularly.

Services already exist for simple performance profiling in CMS, but these profile events such as function calls as being disjoint from anything else in the execution: one example is designed to provide a visualisation of performance, and is not a verification tool.  Therefore, this report lays the foundation for later stages of research, most of which will make use of CMS infrastructure and improve on the existing CMS tools.

\chapter{Control Flow Temporal Logic}\label{section-logic-outline}

This chapter introduces the new logic whose semantics, instrumentation and monitoring are the focus of this report.  The chapter will open in Section \ref{section-static-model} with the introduction of a new static representation of programs, a so-called \textit{Symbolic Control Flow Graph} (SCFG).  This abstract representation of programs will serve as both a basis on which to build the new logic's semantics, and an aid to the instrumentation method described later.  Section \ref{section-program-run-model} will describe how the abstraction in Section \ref{section-static-model} can be combined with a set of \textit{critical variables} (ultimately, an alphabet) to induce a model of program runs.  This model of program runs will be used to define the semantics of the new logic in Section \ref{section-logic-semantics}.

\section{A Static Model of Programs}\label{section-static-model}

Consider a function $f$ implemented in code.  Now, let $\text{Var}_f$ be the set of program variables (including functions defined/called) in the code for the function $f$.  Further, let a \textit{symbolic state}, a representation of program state without runtime information, be a total function $\sigma : \text{Var}_f \to \{\text{changed}, \text{unchanged}, \text{undefined}, \text{called}\}$.  Thus, a symbolic state $\sigma$ determines whether a program variable has been changed or whether a function has been called, but contains no information about the value since it cannot be assumed that this is computable statically.  Now the notion of a symbolic state has been introduced, the Symbolic Control Flow Graph of the code in the function $f$ is given in Definition \ref{def-scfg}.

\begin{definition}[Symbolic Control Flow Graph]\label{def-scfg}
A Symbolic Control Flow Graph of a function $f$ is a directed graph $SCFG_f = \langle V, E, v_s, V_e \rangle$ which allows cycles.  $V, E, v_s$ and $V_e$ are such that

\begin{itemize}

	\item $V$ is a finite sequence of symbolic states $\langle \sigma_1, \sigma_2, \dots, \sigma_n \rangle$ induced by state changing instructions in the code for $f$.
	
	Denote by $V(\sigma)$ the natural number $i$ such that $(V, i) = \sigma$ (where $(V, i)$ is the $i^{\text{th}}$ element of the vector $V$), ie, the index in $V$ at which $\sigma$ is found.  Also, denote by $|V|$ the length of the sequence $V$.
	
	\item $E \subset \mathbb{N}_{|V|} \times \Phi \times 2^{\{\text{call}, \text{assignment}, \text{control flow}\}} \times \mathbb{N}_{|V|}$ is a set of edges representing instructions that induce changes in state.  In particular, $\langle n, \phi, \text{type}, m\rangle \in E$ is an edge from $(V, n)$ to $(V, m)$ augmented with a branching condition $\phi$ (possibly empty) that asserts that $\phi$ holds between the symbolic states $(V, n)$ and $(V, m)$, and that the computation required to move from symbolic state $(V, n)$ to $(V, m)$ is a statement of type $t$ for every $t \in \text{type}$.
	
	\item $v_s \in \mathbb{N}_{|V|}$ is the position in $V$ of the symbolic state from which the control flow starts.
	
	
	\item $V_e \subset \mathbb{N}_{|V|}$ is the set of indices in $V$ of final symbolic states where, for every $i \in V_e$,  there is no $\langle i, \phi, t, m \rangle \in E$ such that $m \in V_e$ (if a symbolic state is final, it cannot be moved to another final symbolic state).
	
\end{itemize}
\end{definition}

Given an edge $\langle n, \phi, \text{type}, m \rangle \in E$, $t \in \text{type}$ is a \textit{type} of the edge.  Further, given $\text{SCFG}_f$, a path $\pi$ of length $l$ (write $|\pi| = l$) through $\text{SCFG}_f$ is defined by a finite sequence of symbolic states $\langle \sigma_1, \dots, \sigma_l \rangle$ where, for each $\sigma_i, \sigma_{i+1}$ ($1 \le i < l$), there is an edge $\langle V(\sigma_i), \phi, t, V(\sigma_{i+1}) \rangle \in E$.  A \textit{complete path} is a path $\pi = \langle \sigma_1, \dots, \sigma_l \rangle$ such that $V(\sigma_1) = v_s$ and $V(\sigma_l) \in V_e$.

\begin{remark}
$v_s$ is normally the empty symbolic state, denoted by $\sigma_\epsilon$.  For $\sigma_\epsilon$, for every $x \in \text{Var}_f$, $\sigma_\epsilon(x) = \text{undefined}$.
\end{remark}

\begin{remark}
A permuted form of $V$, given the appropriately modified edge set $E$, would define an isomorphic graph (the isomorphism would be a permutation of natural numbers\footnote{A permutation of a set is a bijective map from the set to itself.}); the use of a sequence for $V$ is simply to have a natural index to refer to each state and, more significantly, to allow multiple occurrences of the same symbolic state.
\end{remark}

\begin{remark}
Symbolic Control Flow Graphs are necessarily finite ($|V|, |E| < \infty$) because programs are finite representations of algorithms.
\end{remark}

Construction of a Symbolic Control Flow Graph is outlined in the form of a set of rules.  Suppose the current symbolic state is $\sigma$ (hence, a vertex in some SCFG).  Then:

\begin{itemize}
	\item Assignments \lstinline{x = a} for some program variable $x$ and some expression $a$ induce a new state $\sigma'$ with $\sigma' = \sigma[x \mapsto \text{changed}]$\footnote{$\sigma[x \mapsto \text{changed}]$ is the standard map modification notation, denoting the map that agrees with $\sigma$ on all program variables except $x$, whose value is now $\text{changed}$.}.  Assignments therefore also induce edges $\langle \sigma, \top, \text{assignment}, \sigma' \rangle$.
	
	\item Function calls \lstinline{g(...)} for some function $g$ (the internals of which are not assumed to be known and are currently of no interest) induce a new state $\sigma'$ with $\sigma' = \sigma[g \mapsto \text{called}, x_1 \mapsto \text{changed}, \dots, x_n \mapsto \text{changed}]$ for all $x_i \in \text{Var}_f$.  Notice the inclusion of all program variables being changed; this is because one cannot make any assumptions about the purity\footnote{A function is pure if no variables outside the local scope of that function are affected by its execution.} of the function $g$.  Function calls therefore also induce edges $\langle \sigma, \top, \text{call}, \sigma' \rangle$.
	
	\item Conditionals induce multiple new states $\sigma'_i$ for $k$ blocks by considering each block individually: the change in state induced by the first instruction in block $i$ induces the new state $\sigma'_i$.  Naturally, the block reached by the conditional's test condition $\phi_1$ induces an edge $\langle \sigma, \phi_1, t_1, \sigma'_1 \rangle$; the blocks reached by the alternative conditions $\phi_i$ for $1 < i < k$ induce edges
	
	$$\langle \sigma, \bigg(\bigwedge_{n=1}^{i-1} \lnot\phi_n \bigg) \land \phi_i, t_i, \sigma'_i \rangle,$$
	
	and the else-block (the $k^{\text{th}}$) induces an edge
	
	$$\langle \sigma, \bigwedge_{n=1}^{k-1} \lnot \phi_n, t_n, \sigma'_k \rangle.$$
	
	Each $t_i$ for $0 \le i \le n$ is the type of the statement that induces the symbolic state $\sigma'_i$.
	
	\item For and while-loops induce new states by applying the rules above to the code in their body, then taking the first and last symbolic states in the body ($\sigma_s$ and $\sigma_t$ respectively) and inducing edges, based on the loop condition $\phi_l$, $\langle \sigma, \phi_l, t, \sigma_s \rangle$ and $\langle \sigma_t, \top, \text{control flow}, \sigma \rangle$.  Here, $t$ is the type of the statement that induces the symbolic state $\sigma_s$.
\end{itemize}

In the case of the instruction \lstinline{x = e(g(...))} where $e$ denotes an arbitrary expression (hence, an assignment and a function call in one instruction), the rules can simply be combined, where all program variables must be considered and $g$ is considered as called.  Additionally, processing a conditional leads to multiple \textit{loose-end} states (one for each branch), so the next instruction (if there is one) must be processed with respect to each symbolic state that is still a loose-end.  Finally, for a symbolic state that indicates that the variable $x \in \text{Var}_f$ has changed/been called, the next symbolic state after that which indicates that a program variable has been changed/called that is not $x$ must map $x$ to \textit{unchanged}.

\begin{example}
Figure \ref{fig-scfg-example} shows an example SCFG computed from a simple program.  The program, in this case, has branching based on the condition \lstinline{i == j}; the edges labelled \lstinline{[...]} (shortened to save space) denote the conditions that must hold for the control flow to follow the respective paths.

	\begin{figure}[ht]
	\centering
	\begin{subfigure}{0.4\textwidth}
\begin{lstlisting}
a = 10
i = 10
j = 20
if i == j:
	c = 10
	f(0.1)
else:
	c = 20
	f(1.1)
\end{lstlisting}
	\caption{\label{subfig-scfg-example-code}An example code snippet.}
	\end{subfigure}%
	\begin{subfigure}{0.4\textwidth}
	\includegraphics[width=0.6\linewidth]{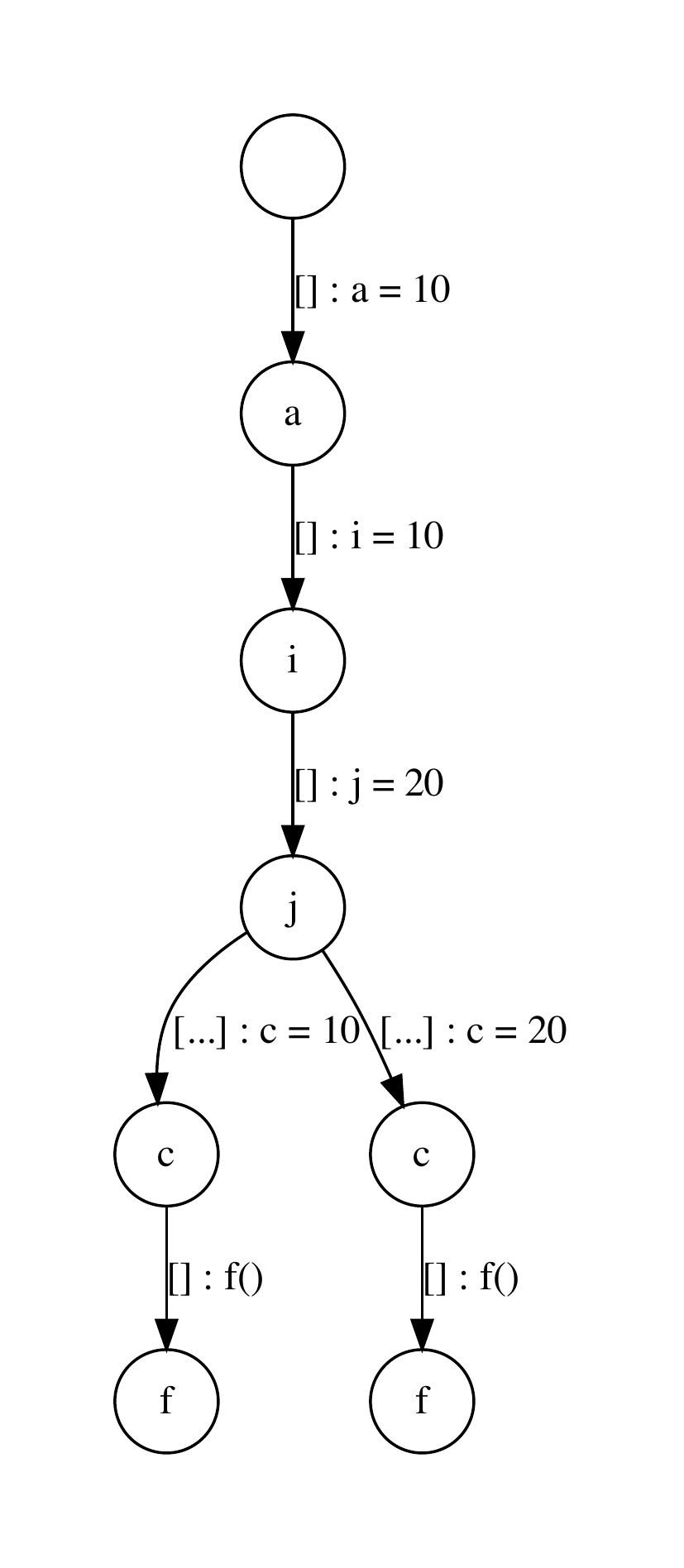}
	\caption{\label{subfig-scfg-example}The SCFG derived from the code snippet.}
	\end{subfigure}
	
	\caption{\label{fig-scfg-example}}
	\end{figure}
\end{example}

\subsection{Similarities with Symbolic Execution Trees}

The symbolic execution tree \cite{Baldoni, Kinga} used in the symbolic execution literature draws similarities with the SCFG approach in that it models a program without any runtime information.  A key difference is that symbolic execution trees associate both symbolic states \textit{and} the instructions that compute the symbolic states with vertices, using edges only for branching.  The SCFG described in Section \ref{section-static-model}, however, only associates symbolic states with vertices, and associates the instructions whose execution computes those symbolic states with the preceding edges (that is, the edges whose destination vertex is the symbolic state computed).

Additionally, despite not using any runtime information, symbolic execution trees encode information about runs of a program and so contain distinct branches that never converge after divergence has been forced by, say, a conditional.  In particular, paths through a symbolic execution tree may be uniquely identified by the sequence of conditions that are true along them.  The symbolic control flow graph, on the other hand, allows convergence since it can be regarded as an augmented control flow graph.

\section{A Model of Program Runs}\label{section-program-run-model}

In order to define a representation of a given program run based on a symbolic control flow graph (see Section \ref{section-static-model}), the set of \textit{critical symbols} is first defined as the set $C$ containing names of program variables and functions, hence $C \subseteq \text{Var}_f$, for the function $f$ implemented in code.  Further, a \textit{concrete state} is a total function

$$\tau : \text{Var}_f \to \text{Val}\cup\{\text{undefined}, \text{not called}\}$$

where $\text{Val}$ is the finite\footnote{Using the intuition that physical machines have finite memory.} set of values that can occur in a program.  The additional values \textit{undefined} and \textit{not called} are added to the codomain to model a variable that is currently undefined, and to model a function that was not called, respectively.  Further, a concrete state models a function call result by storing the return value of the call; if a function has not been called, the value mapped to is \textit{not called}.  Denote by $s\tau$ the symbolic state $\sigma$ that generates $\tau$ at runtime.

Now, consider a run of the function $f$ as a complete path $\pi$ (a path starting at the start state and ending at an end state) through $\text{SCFG}_f$, but with the symbolic states replaced by concrete states holding runtime information.  Definition \ref{def-dds} details the structure of such a run (modelled as a discrete-time dynamical system) with respect to $C$.

\begin{definition}[A Run of $f$ as a Discrete-time Dynamical System]\label{def-dds}

A discrete-time dynamical system constructed from $\text{SCFG}_f = \langle V, E, v_s, V_e \rangle$ using the set $C$ of critical symbols is a tuple $\mathcal{D} = \langle T, \gamma \rangle$ with:

\begin{itemize}
\item $T$ a finite sequence of the form $\langle \tau_1, \tau_2, \dots, \tau_n \rangle$ for concrete states $\tau_i$ such that

\begin{itemize}
	\item $\tau_i$ is derived from some $\sigma_j$ in $V$, hence $s\tau_i = \sigma_j$;
	\item For $\tau_i$ and $\tau_{i+1}$, there is a path $\pi_i$ from $s\tau_i$ to $s\tau_{i+1}$ in $\text{SCFG}_f$; and
	\item For some $x \in C$, \textit{either}
	
	\begin{itemize}
		\item $s\tau_i(x) = \text{changed}$\footnote{If a state changes a value, this can be either because of direct assignment, or as a result of a call to an impure function, hence no restriction is placed on the type of the incident edge.} \textit{or}
		\item if $s\tau_i(x) = \text{called}$, then there is an edge of type $\text{call}$ from $s\tau_{i-1}$ to $s\tau_i$ in $\text{SCFG}_f$.
	\end{itemize}
\end{itemize}

Denote by $|T|$ the length $n$ of the vector and denote by $T(\tau)$ the natural number $i$ such that $(T, i) = \tau$ (where $(T, i)$ is the $i^{\text{th}}$ element of the vector $T$).

\item $\gamma : \mathbb{N}_{|T|} \to \mathbb{R}_{\le}$, the \textit{clock function}, giving the time elapsed in the system $\mathcal{D}$ when the concrete state at position $i$ in $T$ is reached.  In addition, $\gamma(i) = \gamma(j) \iff i = j$ (different concrete states cannot be attained at the same time) and $i < j \iff \gamma(i) < \gamma(j)$ (time must move forward).
\end{itemize}

\end{definition}

With a discrete time dynamical system defined, a natural next step is, for each concrete state $\tau_i$ in the vector $T$ of $\mathcal{D}$, to extend the incident edges of $s\tau_i$ to being a part of the runtime.  Edges in a symbolic control flow graph can have no notion of time since they are statically computed; lifting the notion of an edge to a part of the runtime ultimately allows one to consider timing constraints.  This is done by Definition \ref{def-transition}.

\begin{definition}[Transition]\label{def-transition}
A \textit{transition} $\Delta\tau_i = \tau_i \to \tau_{i+1}$ represents the computation required to move from the concrete state $\tau_i$ to $\tau_{i+1}$ in $T$.  Since the edges $\pi_i$ in the path $\pi$ from $s\tau_i$ to $s\tau_{i+1}$ in $\text{SCFG}_f$ represent instructions, the transition $\Delta\tau_i$ represents the execution of the sequence of instructions in $\pi$ at runtime.
\end{definition}

For a transition $\Delta\tau_i$, denote by $\gamma(\Delta\tau_i)$ the \textit{start time} of the transition where $\gamma(\Delta\tau_i) = \gamma(\tau_i)$.  Hence, the start time of a transition is the time at which the state immediately before it is attained.

From now on, when $\tau_i$ or $\Delta\tau_i$ is written, it is done with the understanding that the timestamp for its occurrence is also available, meaning the pairs $(\tau_i, \gamma(\tau_i))$ and $(\Delta\tau_i, \gamma(\Delta\tau_i))$ are unique\footnote{Without timestamp information, concrete states are simply total functions, and these can be isomorphic; timestamp information breaks the isomorphism and removes ambiguity.}.  Therefore, $\tau_i$ and $\Delta\tau_i$ are notational shortcuts for these pairs.  Furthermore, the containment relation is extended to sequences in the expected way: $\tau_i \in T$ means that there is some $1 \le i \le |T|$ such that $(T, i) = \tau_i$.

\begin{remark}
Transitions can be the computation performed by multiple edges, one after the other; this is from the condition that there must be a \textit{path} between symbolic counterparts, and not just a single edge (though a single edge still constitutes a path, so the symbolic counterparts may be adjacent).
\end{remark}

\begin{remark}
When seen as a function, the operator $s$ that gives the symbolic state $\sigma$ that generates a concrete state $\tau$ holds the following properties:

\begin{itemize}
	\item If there is any branching in $\text{SCFG}_f$ and the section of code with branching is traversed only once, then $s$ cannot be surjective since not every branch is explored at runtime.
	\item If there are any loops in $\text{SCFG}_f$ \textit{and} a symbolic state lies inside a loop that performs more than one iteration, then $s$ cannot be injective since that symbolic state will generate multiple concrete states during runtime.
\end{itemize}
\end{remark}

\begin{remark}
A system $\mathcal{D}$ represents a single run of a function $f$.  If the inputs are changed, or if there is non-determinism in $f$, this results in a different system $\mathcal{D}'$.
\end{remark}

The set of all transitions in a single system $\mathcal{D}$ is denoted by $\Delta\tau$ (with the subscript omitted), but the notation $\mathcal{D} = \langle T, \gamma \rangle$ is maintained since $\Delta\tau$ can be derived from the sequence of concrete states and the times at which they are attained.  The set $\Delta\tau$ has a natural total order $\prec$.  For $\Delta\tau_i, \Delta\tau_j \in \Delta\tau$, $\Delta\tau_i \prec \Delta\tau_j \iff \gamma(\Delta\tau_i) < \gamma(\Delta\tau_j)$.  Using $\prec$, it makes sense to call the minimal element the first element, and build a labelling from there.  This gives a way to index sets of transitions.  Hence, denote by $(\Delta\tau, i)$ the $i^{\text{th}}$ element of $\Delta\tau$ with respect to the ordering $\prec$.

This description of a way to model program runs is concluded by defining some properties of transitions.

\begin{definition}[Properties of Transitions]\label{def-transition-props}
Let $\Delta\tau_i$ be a transition in a system $\mathcal{D}$, with $\Delta\tau_i = \tau_i \to \tau_{i+1}$.  Then, one can define:

\begin{itemize}
	\item $d : \Delta\tau \to \mathbb{R}_{\ge}$ with $d(\Delta\tau_i) = \gamma(\tau_{i+1}) - \gamma(\tau_i)$.  $d(\Delta\tau_i)$ is called the \textit{duration} of $\Delta\tau_i$.
	\item $\text{source}(\Delta\tau_i)$ and $\text{dest}(\Delta\tau_i)$ to be the concrete states $\tau_i$ and $\tau_{i+1}$ respectively.
	\item $\text{incident}(\tau_i)$ to be $\Delta\tau_{i-1}$.
\end{itemize}
\end{definition}

Note that, if a transition $\Delta\tau_i$ corresponds to a function call (that is, it corresponds to a single edge of type \textit{call} in a SCFG), the duration $d(\Delta\tau_i)$ of the transition is the time taken by the function call.  In particular, by the conditions in Definition \ref{def-dds}, if the transition between two concrete states corresponds to a function call, this must be represented by a single edge in the SCFG.  Furthermore, $d(\Delta\tau_i)$ for any transition $\Delta\tau_i$ in a system $\mathcal{D}$ provides a mechanism to directly talk about time constraints, which are simply predicates on the map $d(\Delta\tau_i)$.

\begin{example}\label{eg-dds}
Consider again the code snippet

\begin{lstlisting}
a = 10
for i in range(4):
	if i < 3:
		f(0.1)
	else:
		f(1.1)
\end{lstlisting}

with the SCFG in Figure \ref{subfig-scfg-example}.  In the case of this code, there are no parameters and the program is deterministic so, for each set of critical variables $C$, only one DDS can ever be generated by runs of it.  Suppose $\mathcal{D} = \langle T, \gamma \rangle$ is the DDS with $C = \{a, f\}$.  Then,

\begin{equation*}
\begin{split}
T = & \; ([a \mapsto 10, f \mapsto \text{not called}],
[a \mapsto 10, f \mapsto f(0.1)],\\
& \;[a \mapsto 10, f \mapsto \text{not called}], [a \mapsto 10, f \mapsto f(0.1)],\\
& \;[a \mapsto 10, f \mapsto \text{not called}], [a \mapsto 10, f \mapsto f(0.1)],\\
& \; [a \mapsto 10, f \mapsto \text{not called}], [a \mapsto 10, f \mapsto f(1.1)])
\end{split}
\end{equation*}

and $\gamma$ assigns to each concrete state in the sequence $T$ a timestamp.
\end{example}

This paves the way to defining the new logic that is the main topic of this report.

\section{CFTL and its Semantics}\label{section-logic-semantics}

Given the machinery developed so far, it is now possible to define the new logic; this involves defining the syntax and semantics.  First, some requirements should be given.  The logic should:

\begin{itemize}
	\item Describe constraints over state and time.  The main point of interest for time constraints is over function calls.
	\item Be efficient to check at runtime, meaning reaching a verdict on a property will not generate too much\footnote{``Too much'' depends on the system being monitored.} overhead.
\end{itemize}

With this in mind, the purpose of developing this logic should be reiterated.  Given a function $f$ that is implemented in code, a property $\phi_f$ should be checked with respect to the function $f$.  If the function, at any point, violates this property, a $\bot$ (false) verdict should be reached.  If no violation occurs, the verdict is always either $\top$ (true, so no violation can occur for the section of the code being monitored) or ? (not enough information has been observed yet to reach a true or false verdict).

Formulas in this logic will take the form

\begin{equation}\label{eq-cftl-form}
\phi_f \equiv \forall q_1 \in S_1, \dots, \forall q_n \in S_n : \psi(q_1, \dots, q_n)
\end{equation}

where $S_i$ for $1 \le i \le n$ are domains over which quantification can occur (see Definition \ref{def-qd}) and $\psi(q_1, \dots, q_n)$ is an $n$-ary predicate (a predicate with $q_1, \dots, q_n$ as free variables).  The significant structure of a formula in this logic lies in the predicate $\psi$.  It remains to define the $S_i$; this is done in Definition \ref{def-qd}, but some preliminary work is required to allow the construction of the sets discussed there.

\subsection{Quantification Domains}\label{section-qds}

Let $\mathcal{D} = \langle T, \gamma \rangle$ be a discrete-time dynamical system (see Definition \ref{def-dds}) based on some $\text{SCFG}_f$ with respect to a set of critical symbols $C$.  Now, two binary relations must be defined.  Let:

\begin{itemize}

\item $\text{\underline{is}} \subset \Delta\tau \times 2^{\{\text{assignment}, \text{call}, \text{control flow}\}}$ with $\Delta\tau_i \text{ \underline{is} } t \iff$ the edge to which $\Delta\tau_i$ corresponds in $\text{SCFG}_f$ has type $t$ (if $\Delta\tau \text{ \underline{is} } \{t\}$ for a singleton $\{t\}$, then $\Delta\tau \text{ \underline{is} } t$ is written);

\item $\text{\underline{operates on}} \subset \Delta\tau \times 2^C$ with $\Delta\tau_i \text{ \underline{operates on} } C' \iff \text{for every $x \in C'$, }\text{dest}(\Delta\tau_i)(x) \in \{\text{changed}, \text{called}\}$.  If $\Delta\tau_i \text{ \underline{operates on} } C'$ and $C' = \{x\}$ is a singleton set, for simplicity, one writes $\Delta\tau_i \text{ \underline{operates on} } x$.

\end{itemize}

Now, let $P_{\Delta\tau} : \Delta\tau \times \{\text{assignment}, \text{call}\} \times C \to \{\top, \bot\}$ be a predicate on transitions with

\begin{equation}\label{eq-P-delta}
P_{\Delta\tau}(\Delta\tau_i, t, x) \equiv (\Delta\tau_i \text{ \underline{is} } t) \land (\Delta\tau_i \text{ \underline{operates on} } x)
\end{equation}

If the transition is a function call, the \underline{operates on} relation gives the name of the function being called as well as every $x_i \in C$; if it is an assignment, the relation gives the name of the variable to which it assigns a value.

Now, define a third binary relation $\text{\underline{changes}} \subset \tau \times C$ by $\tau_i \text{ \underline{changes} } x \iff \tau_i(x) \neq \tau_{i-1}(x)$ and so $P_T : T \times C \to \{\top, \bot\}$ is a predicate on concrete states with

\begin{equation}\label{eq-P-tau}
P_T(\tau_i, x) \equiv (\tau_i \text{ \underline{changes} } x).
\end{equation}

With $P_{\Delta\tau}$ and $P_T$ defined, it now makes sense to write $\Delta\tau_i \models P_{\Delta\tau}(\Delta\tau_i, t, x)$, ie, ``the transition $\Delta\tau_i$ holds the property $P_{\Delta\tau}(\Delta\tau_i, t, x)$'', and similarly for $P_T$.  Hence, $\Delta\tau_i \models P_{\Delta\tau}(\Delta\tau_i, t, x) \iff P_{\Delta\tau}(\Delta\tau_i, t, x) \equiv \top$.

Finally, let $\Gamma = P_{\Delta\tau}(\Delta\tau_i, t, x) \text{ or } P_T(\tau_i, x)$ and denote by $S_{\Gamma}$ a set consisting of elements of either $\Delta\tau$ or $T$ such that, $\forall q \in S_{\Gamma}, q \models \Gamma$.  Recall that, by writing $\tau_i$ or $\Delta\tau_i$, the timestamp at which either occurs is understood, allowing sets to contain multiple instances of isomorphic states and transitions (distinguished only by their timestamps).

\begin{example}
One can construct a set of transitions that are function calls of some function $g$ by writing

\begin{gather*}
\Gamma = (\Delta\tau_i \text{ \underline{is} call }) \land (\Delta\tau_i \text{ \underline{operates on} } g), \text{ and so }\\
S_\Gamma = \{\Delta\tau_i \in \Delta\tau : \Delta\tau_i \models \Gamma\}.
\end{gather*}
\end{example}

\begin{example}
One can also construct a set of all states in which $x$ has a new value by writing

\begin{gather*}
\Gamma = \tau_i \text{ \underline{changes} } x, \text{ and so } S_\Gamma = \{\tau_i \in T : \tau_i \models \Gamma\}.
\end{gather*}
\end{example}

The necessary definitions are now in place to properly present the notion of a Quantification Domain and, therefore, make clear what the domains $S_i$ are in the formula Equation \ref{eq-cftl-form}.

\begin{definition}[Quantification Domain (QD)]\label{def-qd}

Let $\Gamma$ be a property over states or transitions, and let $S_\Gamma$ be the set of either states or transitions that hold this property.  Then $S_\Gamma$ is a Quantification Domain (QD).

\end{definition}

As an example, let $S_\Gamma$ be the set of states that change $x$, so $S_\Gamma = \{\tau_i : \tau_i \models (\tau_i \text{ \underline{changes} } x)\}$.  Then, $\phi_f \equiv \forall q \in S_\Gamma : \psi(q)$ (a case of one-dimensional quantification) can be interpreted as ``For every state that changes the program variable $x$, the predicate $\psi(q)$ on that state should hold''.  This a natural way of applying properties to programs and is the idea followed through this report.

To finish this section, the precise definition of the natural ordering with respect to time is needed for quantification domains.  Let $S$ be a quantification domain of either concrete states or transitions taken from the system $\mathcal{D} = \langle T, \gamma \rangle$.  Then the total ordering $\prec$ on $S$ induced by $\gamma$ is such that, for $s_1, s_2 \in S$, $s_1 \prec s_2 \iff \gamma(s_1) < \gamma(s_2)$.

Using this total order, the minimal element is the first element, and a labelling can be applied from there, thus generating an indexing of the elements of a quantification domain.

\subsection{Points of Interest and Future Time}

Given a formula with one-dimensional quantification, $\forall q \in S : \psi(q)$, one can refer to each $q \in S$ as a \textit{point of interest}.  In particular, $\psi(q)$, once defined, will be a predicate on both $q$ \textit{and} other states/transitions in $\mathcal{D}$ that have some relationship to $q$ (eg, the next transition with respect to $q$ that holds some property $\Gamma$ based on the relations described in Section \ref{section-qds}).

The next step is to define a set of functions that, given a \textit{point of interest} $q \in S$ for some quantification domain $S$, will give either a single element or a set of elements that have some relationship to $q$.

\begin{definition}[Future-time Operators]\label{def-future-time-operators}
Let $\mathcal{D} = \langle T, \gamma \rangle$ be a discrete-time dynamical system.  Then,

\begin{itemize}
\item $\text{next}_{\Delta\tau}(q, \phi)$ gives the next transition in time with respect to $q$ that satisfies $\phi$:

\begin{equation*}
\begin{split}
\text{next}_{\Delta\tau}(q, \phi) = & \; \Delta\tau_i \text{ such that } \\
& (\gamma(\Delta\tau_i) > \gamma(q), \Delta\tau_i \models \phi \text{ and } \\
& \not\exists \Delta\tau_j \in \Delta\tau \text{ with } \gamma(q) < \gamma(\Delta\tau_j) < \gamma(\Delta\tau_i) \text{ and } \Delta\tau_j \models \phi).
\end{split}
\end{equation*}

$\text{next}_\tau$ is similar, but for states.

\item $\text{future}_{\Delta\tau}(q, \phi)$ gives all future transitions in time with respect to $q$ that satisfy $\phi$:

$$
\text{future}_{\Delta\tau}(q, \phi) = \{\Delta\tau_i : \gamma(\Delta\tau_i) > \gamma(q) \text{ and } \Delta\tau_i \models \phi\}
$$

$\text{future}_\tau$ is again similar, but for states.

\end{itemize}

\end{definition}

\begin{remark}
For any $q \in S$ for a quantification domain consisting of concrete states or transitions and a predicate $\phi$ written in the form seen in Section \ref{section-qds}, $\text{next}_{\Delta\tau}(q, \phi) \in \text{future}_{\Delta\tau}(q, \phi)$ and similarly for the future-time operators that give states.
\end{remark}

Notice that, since $\text{next}_{\Delta\tau}$ or $\text{next}_{\tau}$ yield single elements with respect to a system $\mathcal{D}$, these cannot be quantified over; it makes no sense.  However, $\text{future}_{\Delta\tau}$ and $\text{future}_{\tau}$ can indeed be quantified over, since they yield sets.

\begin{example}
Consider the DDS in Example \ref{eg-dds} with the sequence of states

\begin{equation*}
\begin{split}
T = & \; ([a \mapsto 10, f \mapsto \text{not called}],
[a \mapsto 10, f \mapsto f(0.1)],\\
& \;[a \mapsto 10, f \mapsto \text{not called}], [a \mapsto 10, f \mapsto f(0.1)],\\
& \;[a \mapsto 10, f \mapsto \text{not called}], [a \mapsto 10, f \mapsto f(0.1)],\\
& \; [a \mapsto 10, f \mapsto \text{not called}], [a \mapsto 10, f \mapsto f(1.1)]).
\end{split}
\end{equation*}

Fix $q = [a \mapsto 10, f \mapsto \text{not called}]$ (the first entry in $T$) and $\phi \equiv (\Delta\tau_i \text{ \underline{is} call }) \land (\Delta\tau_i \text{ \underline{operates on} } f)$.  Then $\text{next}_{\Delta\tau}(q, \phi)$ (or $\text{next}_{\Delta\tau}(q, f)$ when the type of the transition is understood) refers to the transition

$$
[a \mapsto 10, f \mapsto \text{not called}] \to [a \mapsto 10, f \mapsto f(0.1)]
$$

where the notation used is consistent with that in Definition \ref{def-transition}.
\end{example}

Now, the structure of $\psi(q_1, \dots, q_n)$ is presented in Definition \ref{def-psi}.

\begin{definition}[Form of $\psi(q_1, \dots, q_n)$]\label{def-psi}

The form of $\psi$ is given by the grammar

\begin{equation*}
\begin{split}
\psi := & \; \forall q \in S : \psi_2\\
\psi_2 := & \; \forall q' \in S'(q) : \psi_2 \mid \psi_3\\
\psi_3 := & \; \top \mid p(q) \mid \lnot\psi_3 \mid \psi_3 \lor \psi_3\\
\end{split}
\end{equation*}

Note that this grammar can only generate formulas in prenex normal form.  Note also that implication and conjunction can be expressed by the usual identifies\footnote{$p \implies q \equiv \lnot p \lor q$ and $p \land q \equiv \lnot(\lnot p \lor \lnot q).$}.  Here, $p$ follows the context-sensitive grammar below for some $x \in C$; $q$ is the current binding from a quantification domain; and $q' \in S'(q)$ denotes nested quantification with respect to a quantification domain whose computation requires a binding $q$ from some $S$.

\begin{equation*}
\begin{split}
p(q : \text{\textbf{state}}) := & \; q(x) = n \mid q(x) \in (n, m) \mid q(x) \in [n, m] \mid p(\text{incident}(q))\\ 
p(q: \text{\textbf{transition}}) := & \; d(q) \in [n, m] \mid d(q) \in (n, m) \mid p(\text{source}(q)) \mid p(\text{dest}(q))\\
p(q : \text{\textbf{state}} \text{ or } \text{\textbf{transition}}) := & \; p(\text{next}_{\Delta\tau}(\zeta)) \mid p(\text{next}_{\tau}(\eta))
\end{split}
\end{equation*}

where $\zeta$ and $\eta$ take the form defined in Equations \ref{eq-P-delta} and \ref{eq-P-tau}, respectively and $m$ and $n$ are natural numbers fixed when the formula is written.  

\end{definition}

It is assumed in the grammar that formulas have at most as many free variables as there are bound variables from quantification; there can be no free variables that are not bound by some quantifier.

Now, the semantics can be given.  Some notation is introduced, first:  $S(q)$, for a quantification domain $S$ depending on an existing binding $q$, denotes the instance of $S$ when the binding $q$ is given.  For a formula $\phi_f \equiv \forall q_1 \in S_1, \dots, \forall q_n \in S_n : \psi(q_1, \dots, q_n)$ with quantification sequence $\forall q_1 \in S_1, \dots, \forall q_n \in S_n$, a binding $\beta$ is a map from bind variables to $T \cup \Delta\tau$ (for $T$ and $\Delta\tau$ part of some DDS) derived from the quantification sequence.  Note that bindings may be partial functions.

\begin{definition}[Definition of $\beta \models \phi_f$ for a system $\mathcal{D}$]\label{def-semantics}

Let $\mathcal{D} = \langle T, \gamma \rangle$ be a discrete-time dynamical system, let $\phi_f \equiv \forall q_1 \in S_1, \dots, \forall q_n \in S_n : \psi(q_1, \dots, q_n)$ be a property in CFTL and let $\beta$ be a binding taken from the quantification sequence $\forall q_1 \in S_1, \dots \forall q_n \in S_n$.  Then, the relation $\beta \models \psi(q_1, \dots, q_n)$ is defined by:

\begin{equation*}
\begin{split}
\beta \models \top\\
\beta \models \phi_1(q_1, \dots, q_n) \lor \phi_2(q_1, \dots, q_n) \iff & \; \beta \models \phi_1(q_1, \dots, q_n) \text{ or } \beta \models \phi_2(q_1, \dots, q_n)\\
\beta \models \lnot\phi(q_1, \dots, q_n) \iff & \; \beta \not\models \phi(q_1, \dots, q_n)\\
\end{split}
\end{equation*}

Now, take a state $s$ from a binding $\beta$ (that is, $\beta(q_i) = s$ for some $i$).  Then, the semantics for $s \models p(q)$ (where $p$ has the structure given in Definition \ref{def-psi}) follows.

\begin{equation*}
\begin{split}
s \models q(u) = n \iff & \; s(u) = n\\
s \models q(u) \in (n, m) \iff & \; s(u) \in (n, m)\\
s \models q(u) \in [n, m] \iff & \; s(u) \in [n, m]\\
s \models p(\text{incident}(q)) \iff & \; p(\text{incident}(s)) \text{ holds}\\
\end{split}
\end{equation*}

Suppose now that $t$ is a transition taken from a binding $\beta$ (that is, $\beta(q_i) = t$ for some $i$).  Then, the semantics for $t \models p(q)$ (where $p$ has the structure given in Definition \ref{def-psi}) follows.

\begin{equation*}
\begin{split}
t \models d(q) \in (n, m) \iff & \; d(t) \in (n, m)\\
t \models d(q) \in [n, m] \iff & \; d(t) \in [n, m]\\
t \models p(\text{source}(q)) \iff & \; p(\text{source}(t)) \text{ holds}\\
t \models p(\text{dest}(q)) \iff & \; p(\text{dest}(t)) \text{ holds}\\
\end{split}
\end{equation*}

Supposing that $e$ is either a transition \textit{or} a state.  Then the remaining semantics is:

\begin{equation*}
\begin{split}
e \models p(\text{next}_{\Delta\tau}(\zeta)) \iff & \; \text{there is } \Delta\tau_j \text{ such that:}\\
& \; (\text{there is no } \Delta\tau_k \in \Delta\tau\text{ with $\gamma(e) < \gamma(\Delta\tau_k) < \gamma(\Delta\tau_j)$}\\
& \;\;\; \text{ and }\Delta\tau_k \models \zeta)\\
& \; \text{ and  } p((\Delta\tau, j))  \text{ holds}\\
e \models p(\text{next}_{\tau}(\zeta)) \iff & \; \text{there is } j \text{ such that:}\\
& \; (\text{there is no } \tau_k \in \tau \text{ with $\gamma(e) < \gamma(\tau_k) < \gamma((T, j))$}\\
& \;\;\; \text{ and }\tau_k \models \zeta)\\
& \; \text{ and  } p((T, j)) \text{ holds}\\
\end{split}
\end{equation*}

\end{definition}

\begin{remark}
In $\phi_f \equiv \forall q_1 \in S_1, \dots, \forall q_n \in S_n : \psi(q_1, \dots, q_n)$, the first quantification domain $S_1$ used is necessarily \textit{independent} of any bindings, that is, it does not require any binding from any other quantification domain to be computed.  All other quantification domains $S_i$ are necessarily \textit{dependent} on some $S_j$ with $j < i$.
\end{remark}

Definition \ref{def-semantics} gives the notion of a binding $\beta$ satisfying a formula $\psi(q_1, \dots, q_n)$.  This definition is now extended to say what it means for a system $\mathcal{D}$ to hold the property $\phi_f$, written $\mathcal{D} \models \phi_f$.  Denote by $\textsc{Obs} = (s_1, \dots, s_k)$ the \textit{current observation sequence} such that $s_i \in T \cup \Delta\tau$ for some DDS $\mathcal{D} = \langle T, \gamma \rangle$.  An observation sequence is intuitively the data observed so far from a runtime; some $\textsc{Obs}$ is said to be \textit{well-formed} when, for every $1 \le i < k$, $\gamma(s_{i+1}) > \gamma(s_i)$.  Consequently, only well-formed observation sequences are considered as these are the only ones that can be received from the runtime of a monitored program.

Now, consider a formula $\phi_f \equiv \forall q_1 \in S_1, \dots, \forall q_n \in S_n : \psi(q_1, \dots, q_n)$.  The quantification sequence of this formula defines a set $B^*$ of bindings $\beta$ that map the bind variables to elements of $T \cup \Delta\tau$ of some DDS.  Given a current observation sequence $\textsc{Obs}$, a subset of $B^*$ can be said to be generated.  Denote such a subset by $B^*(\textsc{Obs})$ (such a subset can contain maps that are partial, since some information required to construct full bindings may not have been observed, yet).  Now, $\textsc{Obs}$ is a finite prefix of another finite sequence $\textsc{Obs}^*$ that represents the sequence of observations obtained by observing the entire runtime, hence it is natural to consider \textit{extensions} of such a prefix, say $\textsc{Obs}'$, and consider the concatenation of sequences, $\textsc{Obs} + \textsc{Obs}'$.  One can then write $\textsc{Obs} + \textsc{Obs}' = \textsc{Obs}^*$.

The 3-valued semantics of a formula $\phi_f$ based on a finite prefix of an observation sequence can therefore be given in terms of these finite prefixes, as it is in Definition \ref{def-3-valued-semantics}.  This definition gives a value to $[\textsc{Obs} \models \phi_f]$, which denotes the verdict of $\phi_f$ given the observation sequence (a finite prefix) $\textsc{Obs}$.

\begin{definition}[3-valued semantics of $\phi_f$]\label{def-3-valued-semantics}
Let $\phi_f$ be a formula in CFTL and $\textsc{Obs}$ be the current observation sequence.  Then,

$$
[\textsc{Obs} \models \phi_f] =
\begin{cases}
\top & \text{if for every } \textsc{Obs}', |B^*(\textsc{Obs} + \textsc{Obs}')| = |B^*(\textsc{Obs})|\\
& \; \text{ and for every } \beta \in B^*(\textsc{Obs}), \beta \models \psi\\
\bot & \text{ if there is } \beta \in B^*(\textsc{Obs}) : \beta \not\models \psi\\
? & \text{ if } (\text{there is } \beta \in B^*(\textsc{Obs}) : \psi(\beta) \text{ is neither } \top \text{ nor } \bot) \text{ or }\\
& \; (\text{there is } \textsc{Obs}' : |B^*(\textsc{Obs} + \textsc{Obs}')| > |B^*(\textsc{Obs})|).
\end{cases}
$$
\end{definition}

The intuition is as such:

\begin{itemize}
	\item $[\textsc{Obs} \models \phi_f] = \top$ when no extension of the current sequence of observations can introduce a new binding and, for every existing binding $\beta \in B^*(\textsc{Obs})$, $\psi$ holds.
	
	\item $[\textsc{Obs} \models \phi_f] = \bot$ when, given the current sequence of observations, there is already a binding under which $\psi$ does not hold; observation of more data and expansion of $B^*$ cannot change this.
	
	\item $[\textsc{Obs} \models \phi_f] = \; ?$ when:
	
	\begin{enumerate}
	\item There is a binding $\beta \in B^*$ that does not contain enough information to decide $\psi(\beta)$, where $\psi(\beta)$ denotes $\psi$ with values substituted in from $\beta$; or
	\item There is a possibility that the remainder of the runtime of the program under scrutiny will generate more bindings, which could give rise to a violation of $\psi$.  However, there may also be extensions of the current observation sequence that do not violate $\psi$; one cannot know which extension will be observed.
	\end{enumerate}
\end{itemize}

Finally, for a DDS $\mathcal{D} = \langle T, \gamma \rangle$, write $\mathcal{D} \models \phi_f$ if and only if there is some finite prefix $\textsc{Obs}$ of the complete observation sequence of $\mathcal{D}$, $\textsc{Obs}^*$, such that $[\textsc{Obs} \models \phi_f] = \top$.  Notice that this means that all other finite prefixes with $\textsc{Obs}$ as a prefix will also be models for $\phi_f$ by virtue of Definition \ref{def-3-valued-semantics}.

\begin{remark}
The space of bindings $B^*(\textsc{Obs})$ derived from an observation sequence may contain partial bindings and still contain enough information to decide $\psi$ for each binding.
\end{remark}

From now on, for a binding $B = [q_1 \mapsto v_1, \dots, q_n \mapsto v_n]$, notation will be reduced to $(v_1, \dots, v_n)$ with the understanding that the bind variables in a formula have the same order as the one in which they are written.

\begin{example}\label{eg-obs-sequence}
Consider again the code snippet:

\begin{lstlisting}
a = 10
for i in range(4):
	if i < 3:
		f(0.1)
	else:
		f(1.1)
\end{lstlisting}

and verification of the property

\begin{equation}\label{eq-example-obs}
\phi \equiv \forall q \in S_a : q(a) = 10 \implies d(\text{next}_{\Delta\tau}(q, f)) \in [0, 1].
\end{equation}

Suppose that the runtime has been observed up to a point where the observation sequence is 

\begin{equation*}
\begin{split}
\textsc{Obs} = & \; ([a \mapsto 10, f \mapsto \text{not called}],\\
& \; [a \mapsto 10, f \mapsto \text{not called}] \to [a \mapsto 10, f \mapsto f(0.1)],\\
& \; [a \mapsto 10, f \mapsto f(0.1)]).
\end{split}
\end{equation*}

From this observation sequence, $B^*(\textsc{Obs}) = \{([a \mapsto 10, f \mapsto \text{not called}])\}$; the quantification sequence in the formula in Equation \ref{eq-example-obs} consists of a single bind variable, hence $([a \mapsto 10, f \mapsto \text{not called}])$ is a full binding.  Since the set $B^*(\textsc{Obs})$ is the set of all bindings that can be derived based on the quantification sequence in the formula being checked, and further observation will not yield a larger set of bindings, a true verdict is reached.
\end{example}

\begin{example}
Consider now the same code as in Example \ref{eg-obs-sequence}, but with verification of the property

\begin{equation}\label{eq-example-obs-future}
\phi \equiv \forall q \in S_a, \forall t \in \text{future}_{\Delta\tau}(q, f) : q(a) = 10 \implies d(t) \in [0, 1].
\end{equation}

Suppose that the runtime has, again, been observed up to a point where the observation sequence is

\begin{equation*}
\begin{split}
\textsc{Obs} = & \; ([a \mapsto 10, f \mapsto \text{not called}],\\
& \; [a \mapsto 10, f \mapsto \text{not called}] \to [a \mapsto 10, f \mapsto f(0.1)],\\
& \; [a \mapsto 10, f \mapsto f(0.1)]).
\end{split}
\end{equation*}

Then $B^*(\textsc{Obs}) = \{([a \mapsto 10, f \mapsto \text{not called}], [a \mapsto 10, f \mapsto \text{not called}] \to [a \mapsto 10, f \mapsto f(0.1)])\}$, where the second bind variable is a transition, since it is obtained from a binding from the set $\text{future}_{\Delta\tau}(q, f)$ in the quantifier sequence.

Notice that extending the observation sequence:

\begin{equation*}
\begin{split}
\textsc{Obs} = & \; ([a \mapsto 10, f \mapsto \text{not called}],\\
& \; [a \mapsto 10, f \mapsto \text{not called}] \to [a \mapsto 10, f \mapsto f(0.1)],\\
& \; [a \mapsto 10, f \mapsto f(0.1)],\\
& \; [a \mapsto 10, f \mapsto \text{not called}],\\
& \; [a \mapsto 10, f \mapsto \text{not called}] \to [a \mapsto 10, f \mapsto f(0.1)],\\
& \; [a \mapsto 10, f \mapsto f(0.1)]).
\end{split}
\end{equation*}

yields the set of bindings

\begin{equation*}
\begin{split}
B^*(\textsc{Obs}) = & \; \{([a \mapsto 10, f \mapsto \text{not called}], [a \mapsto 10, f \mapsto \text{not called}] \to [a \mapsto 10, f \mapsto f(0.1)]),\\
& \; ([a \mapsto 10, f \mapsto \text{not called}], [a \mapsto 10, f \mapsto \text{not called}] \to [a \mapsto 10, f \mapsto f(0.1)])\}
\end{split}
\end{equation*}

where the bindings are distinguished by the fact that, when written down, concrete states and transitions are understood to be paired with their timestamps.

To finish this example, notice that observing further calls to $f$ would yield further bindings, and so would expand $B^*(\textsc{Obs})$.  By Definition \ref{def-3-valued-semantics}, this means $[\textsc{Obs} \models \phi]$ cannot be equal to $\top$; rather it must be the case that $[\textsc{Obs} \models \phi] = \;?$ since extensions $\textsc{Obs}'$ of $\textsc{Obs}$ generate larger spaces of bindings $B^*(\textsc{Obs} + \textsc{Obs}')$.

\end{example}

It remains to define the structure of CFTL monitors (Chapter \ref{chapter-monitor-synthesis}), and how these are used in the instrumentation and monitoring algorithms (Chapters \ref{chapter-instrumentation} and \ref{chapter-cftl-monitoring}).

\chapter{Monitor Synthesis}\label{chapter-monitor-synthesis}

Now the semantics (see Definition \ref{def-semantics}) have been developed, it remains to build a mechanism with which to perform online verification of the DDS $\mathcal{D}$ (hence, the function run that it models) against some property.  The method explored in this chapter expresses the formula $\psi(q_1, \dots, q_n)$ in $\forall q_1 \in S_1, \dots, \forall q_n \in S_n : \psi(q_1, \dots, q_n)$ as a tree, which is progressively collapsed as more information needed to come to a verdict about the property is observed.

\section{Online Monitors}

A monitor constructed for online use should have three states, $\{\top, \bot, ?\}$, in agreement with the CFTL semantics in Definition \ref{def-3-valued-semantics}.  The truth value of a formula with respect to observed data must, in the context of online monitoring, be able to be ``I don't know''.

These three states correspond to three distinct configurations of a formula tree (defined in Definition \ref{def-formula-tree}): the root collapsed to $\top$, meaning the truth value is $\top$; the root collapsed to $\bot$, meaning the truth value is $\bot$; and the root not collapsed (still a tree), meaning the truth value is ?.

The definitions that follow introduce the inductively-defined formula tree, along with the notion of collapse of subtrees of formula trees.  This initial definition of a formula tree will be in terms of propositional logic, and is easy to extend to CFTL.

\begin{definition}[Formula Tree]\label{def-formula-tree}
Let $\psi$ be a formula in propositional logic.  Then the Formula Tree $T_{\psi}$ is a directed graph $(V, E, V^*)$ where:

\begin{itemize}
\item $V$ is a set of vertices corresponding to sub-formulas of $\phi$;
\item $E$ is a set of edges where $(\phi, \phi') \in E \iff$ $\phi'$ is a sub-formula of $\phi$;
\item $V^* \in V$ is the root vertex.
\end{itemize}

A tree is defined inductively:

\begin{itemize}

\item $\bigwedge_{i=1}^{n}\phi_i$ generates vertices $\{\bigwedge_{i=1}^{n}\phi_i, \phi_1, \dots, \phi_n\}$ and edges

$$\Bigg\{\Bigg(\bigwedge_{i=1}^{n}\phi_i, \phi_1\Bigg), \dots, \Bigg(\bigwedge_{i=1}^{n}\phi_i, \phi_n\Bigg)\Bigg\}.$$

\item $\bigvee_{i=1}^{n}\phi_i$ generates vertices $\{\bigvee_{i=1}^{n}\phi_i, \phi_1, \dots, \phi_n\}$ and edges

$$\Bigg\{\Bigg(\bigvee_{i=1}^{n}\phi_i, \phi_1\Bigg), \dots, \Bigg(\bigvee_{i=1}^{n}\phi_i, \phi_n\Bigg)\Bigg\}.$$

\end{itemize}

If any $\phi_i$ generated is itself non-atomic (is a conjunction or disjunction; negations are counted as atoms), then the rules for trees generated by conjunctions/disjunctions are applied again.
\end{definition}

A slightly abridged formula tree for $(p \land q) \lor r$ constructed using Definition \ref{def-formula-tree} is given in Figure \ref{fig-formula-tree}.

\begin{figure}[ht]
\centering
\includegraphics[width=0.3\linewidth]{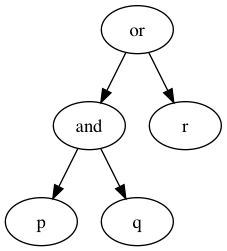}
\caption{\label{fig-formula-tree}The formula tree for $(p \land q) \lor r$.}
\end{figure}

\begin{definition}[Sub-tree Collapse]\label{def-subtree-collapse}
Let $T_{\psi} = (V, E, V^*)$ be the formula tree constructed using Definition \ref{def-formula-tree}, and let $T_{\psi'} = (V', E', (V')^*)$ be a sub-tree of $T_{\psi}$ for a sub-formula $\psi'$ of $\psi$.  Then, the collapse of $T_{\psi'}$ to a truth value $p \in \{\top, \bot\}$ is performed by replacing all edges $(u, (V')^*) \in E$ with $(u, p)$.
\end{definition}

A sub-tree $T_{\psi'}$ can be collapsed provided that one of the following conditions holds:

\begin{itemize}
\item $(V')^* \equiv \phi_1 \lor \dots \lor \phi_n$ for some sub-formulas $\phi_i$ (corresponding to sub-trees) and there is at least one $\phi_i \equiv \top$ (by observation of an atom, or by collapse of sub-trees).  In this case, the collapse is to $\top$.
\item $(V')^* \equiv \phi_1 \land \dots \land \phi_n$ for some sub-formulas $\phi_i$ (corresponding to sub-trees) and all $\phi_i \equiv \top$.  In this case, the collapse is to $\top$.
\item $(V')^* \equiv \phi_1 \land \dots \land \phi_n$ for some sub-formulas $\phi_i$ (corresponding to sub-trees) and there is some $\phi_i = \bot$.  In this case, the collapse is to $\bot$.
\item $(V')^*$ is either a conjunction or disjunction, and all $\phi_i \equiv \bot$.  In this case, the collapse is to $\bot$.
\end{itemize}

\section{Formula Trees extended to CFTL}

Let $\phi_f \equiv \forall q_1 \in S_1, \dots, q_n \in S_n : \psi(q_1, \dots, q_n)$ be a formula in CFTL.  Hence, $\psi(q_1, \dots, q_n)$ consists only of results of predicates joined together with the standard propositional connectives (see Definition \ref{def-psi}).  Then, since the atoms in a CFTL formula are results of predicates applied to concrete states or transitions, deriving the formula tree of $\psi(q_1, \dots, q_n)$ is simply a matter of identifying the predicates and considering those as atoms in the construction described in Definition \ref{def-formula-tree}.  Figure \ref{fig-cftl-formula-tree} illustrates the formula tree computed for the formula

$$
\phi_f \equiv \forall q \in S_a : q(a) \in (0, 10) \implies d(\text{next}(q, f)) \in (0, 1).
$$

Notice that there is no notion of quantification in the formula tree; this is dealt with in the monitoring algorithm presented in Chapter \ref{chapter-cftl-monitoring}

\begin{figure}

\centering

\begin{tikzpicture}[>=latex,line join=bevel,]
  \pgfsetlinewidth{1bp}
\pgfsetcolor{black}
  \draw [->] (92.169bp,72.202bp) .. controls (84.848bp,63.418bp) and (75.893bp,52.672bp)  .. (61.454bp,35.345bp);
  \draw [->] (122.08bp,72.202bp) .. controls (129.47bp,63.483bp) and (138.49bp,52.831bp)  .. (153.09bp,35.593bp);
\begin{scope}
  \definecolor{strokecol}{rgb}{0.0,0.0,0.0};
  \pgfsetstrokecolor{strokecol}
  \draw (168.0bp,18.0bp) node {$d(\text{next}(q, f)) \in (0, 1)$};
\end{scope}
\begin{scope}
  \definecolor{strokecol}{rgb}{0.0,0.0,0.0};
  \pgfsetstrokecolor{strokecol}
  \draw (47.0bp,18.0bp) node {$\lnot(q(a) \in (0, 10))$};
\end{scope}
\begin{scope}
  \definecolor{strokecol}{rgb}{0.0,0.0,0.0};
  \pgfsetstrokecolor{strokecol}
  \draw (107.0bp,90.0bp) node {$\lnot(q(a) \in (0, 10)) \lor d(\text{next}(q, f)) \in (0, 1)$};
\end{scope}
\end{tikzpicture}

\caption{\label{fig-cftl-formula-tree}}
\end{figure}

\section{Formula Trees as Monitors}

Let $\phi$ be a formula in CFTL.  Then it has the form $\phi_f \equiv \forall q_1 \in S_1, \dots, \forall q_n \in S_n : \psi(q_1, \dots, q_n)$.  Now, let $M_{\phi_f}$ be a \textit{monitor}, that is, a mechanism derived from $\psi(q)$ that will report violations when a system $\mathcal{D}$ does not satisfy $\psi(q_1, \dots, q_n)$ for some binding $q_1, \dots, q_n$ as soon after the information required to conclude this is available, and will do nothing when there is no violation.  Definition \ref{def-monitor} introduces $M_{\phi_f}$ formally.

\begin{definition}[Monitor $M_{\psi}$]\label{def-monitor}
If $\phi_f \equiv \forall q_1 \in S_1, \dots, q_n \in S_n : \psi(q_1, \dots, q_n)$, let $T_{\psi}$ be the formula tree for $\psi(q_1, \dots, q_n)$.  Then the monitor $M_{\psi}$ derived from $T_{\psi}$ (and, therefore, from $\psi$) is the tuple $(T_{\psi}, R)$ where $R$ is a map that gives the result/verdict of the monitor from the value of the root vertex of $T_{\psi}$:

$$
R(T_{\psi}) = \begin{cases*}
\top & if $(T_{\psi})^* \equiv \top$\\
\bot & if $(T_{\psi})^* \equiv \bot$\\
? & otherwise
\end{cases*}
$$
\end{definition}

\begin{example}\label{example-monitor-collapse}
Let $\phi \equiv (p \lor q) \land \lnot(r \lor s)$ and let $\alpha = [p \mapsto \top, q \mapsto \top, r \mapsto \bot, s \mapsto \bot]$ be an assignment of the atoms in $\phi$ to truth values.  Based on the collapsation rules in Definition \ref{def-subtree-collapse}:

\begin{itemize}
\item $(p \lor q)$ can be collapsed to $\top$, since $p \equiv q \equiv \top$ in $\alpha$.  The resulting formula is $\phi \equiv \top \land \lnot(r \lor s)$.
\item $\lnot(r \lor s)$ can be rewritten, by DeMorgan's laws, as $\lnot r \land \lnot s$, at which point $\lnot r \land \lnot s$ can be collapsed to $\top$, reducing $\phi$ to $\top$.
\end{itemize}
\end{example}

In Example \ref{example-monitor-collapse}, the assignment $\alpha$ is assumed to be complete.  In the context of online monitoring, an assignment $\alpha$ is constructed as the atoms in $\text{dom}(\alpha)$ are observed to be true or false.

The algorithms in this section will, therefore, compute a truth value in $\{\top, \bot, ?\}$ based on an observed atom, where an atom $p$ is true if and only if $p$ is observed, and $\lnot p$ is not observed.

\section{Algorithms for Checking Formula Tree Truth Values}

Some preliminary definitions are required to make the algorithms present in this section more concise.  Definition \ref{def-formula-closure} gives the set of all sub-formulas of a formula written using the propositional connectives assumed so far.  The remaining two definitions, \ref{def-limited-closure} and \ref{def-limited-m-closure}, amend Definition \ref{def-formula-closure}.  The overall effect is to simplify the presentation of the monitor collapse algorithms presented later.

\begin{definition}[Formula Closure]\label{def-formula-closure}
Let $\phi$ be a formula in the propositional logic.  The closure of $\phi$, denoted by $\text{cl}(\phi)$, is the set of all sub-formulas of $\phi$ with respect to the logical connectives $\lor, \land$.  $\text{cl}(\phi)$ is defined inductively as:

\begin{itemize}

\item $\text{cl}(\bigvee_{i=1}^{n}\phi_i) = \{\bigvee_{i=1}^{n}\phi_i\} \cup \Bigg(\bigcup_{i=1}^{n} \text{cl}(\phi_i)\Bigg)$,

\item $\text{cl}(\bigwedge_{i=1}^{n}\phi_i) = \{\bigwedge_{i=1}^{n}\phi_i\} \cup \Bigg(\bigcup_{i=1}^{n} \text{cl}(\phi_i)\Bigg)$,

\item $\text{cl}(\lnot \phi) = \{\lnot \phi, \phi\} \cup \text{cl}(\phi)$,

\item $\text{cl}(p) = \{p\}$ for an atom $p$.
\end{itemize}
\end{definition}

\begin{remark}
The closure is not defined with respect to negation because, in implementation, DeMorgan's laws are applied to propagate any negation of non-atomic sub-formulas inside, to the atoms.  For example, $\lnot(p \land q)$ is rewritten as $\lnot p \lor \lnot q$.  This ultimately makes the algorithms simpler.
\end{remark}

\begin{example}\label{example-closure}
Let $\phi \equiv (p \land q) \lor r$ be a formula in the propositional logic.  Then $\text{cl}(\phi) = \{(p \land q) \lor r, p \land q, p, q, r\}$.  Since the semantics of the logic in Section \ref{section-logic-outline} uses the propositional logic connectives, Definition \ref{def-formula-closure} is extended to it easily.
\end{example}

\begin{example}\label{example-closure-new-logic}
Let $\phi \equiv \forall s \in S_\Gamma : s(x) \in (0, 10) \implies d(\text{next}_{\Delta\tau}(\alpha)) \in (0, 100)$.  Then

\begin{equation*}
\begin{split}
\text{cl}(\psi(q)) = & \; \{s(x) \in (0, 10) \implies d(\text{next}_{\Delta\tau}(\alpha)) \in (0, 100),\\
& \;\lnot(s(x) \in (0, 10)), s(x) \in (0, 10), d(\text{next}_{\Delta\tau}(\alpha)) \in (0, 100)\}\\
\end{split}
\end{equation*}

based on the identity $p \implies q \equiv \lnot p \lor q$.
\end{example}

\begin{remark}
The closure $\text{cl}(\phi)$ of a formula $\phi$ is analogous to the power set $\mathcal{P}(S)$ of some finite set $S$.
\end{remark}

\begin{definition}[Limited Formula Closure]\label{def-limited-closure}
The limited closure of a propositional formula $\phi$, $\text{l-cl}(\phi)$, is defined, as in Definition \ref{def-formula-closure}, inductively:

\begin{itemize}

\item $\text{l-cl}(\bigvee_{i=1}^{n}\phi_i) = \{\phi_1, \dots, \phi_n\},$

\item $\text{l-cl}(\bigwedge_{i=1}^{n}\phi_i) = \{\phi_1, \dots, \phi_n\},$

\item $\text{l-cl}(p) = \emptyset$ for an atom $p$.
\end{itemize}
\end{definition}

\begin{definition}[Limited Formula Closure as a Multiset]\label{def-limited-m-closure}
As in Definition \ref{def-limited-closure}, but where the sets used are multisets, hence repetitions are allowed.  A limited closure with multisets of a formula $\phi$ is denoted by $\text{l-cl-m}(\phi)$.
\end{definition}

\begin{remark}
Taking the cardinality of the set defined in Definition \ref{def-limited-m-closure} gives the arity of the formula.
\end{remark}

\begin{remark}
Definitions \ref{def-limited-closure} and \ref{def-limited-m-closure} give the operands of a formula, since their definitions do not include recursion.
\end{remark}

\subsection{Recursive Traversal}

Let $\alpha$ be the current assignment for a formula $\phi$ whose truth value is required.  Let $\alpha' = \alpha[p \mapsto b]$ for $b \in \{\top, \bot\}$ (using the map modification notation from Section \ref{section-static-model}).  This section presents two algorithms whose purpose is to take as input an atom, say $p$, and update a formula tree.  Algorithm \ref{alg-recursive-traversal} takes a naive approach without optimisation, and Algorithm \ref{alg-optimised-check} implements a simple optimisation to make finding the points in the formula tree to replace with truth values faster.

Algorithm \ref{alg-recursive-traversal} assumes $\alpha$ (in that the formula $\phi$ already has any $p \in \text{dom}(\alpha)$ replaced by $\alpha(p)$) and computes the new truth value of $\phi$ with $\alpha'$, which is $\alpha$ with the addition of a newly observed atom.

It does this simply by recursing down the formula tree $T_{\psi}$ and collapsing (by Definition \ref{def-subtree-collapse}) vertices if any of the collapse conditions are fulfilled.  Before a bound on the complexity of lookup in such an algorithm is given, the sub-formula relation must be given and is as such:  For formulas $\phi_1$ and $\phi_2$, $\phi_1 \sqsubset \phi_2 \iff$ $\phi_1 \in \text{cl}(\phi_2)\backslash\{\phi_2\}$.  Now, only the complexity of the lookup phase (the phase in which the vertices of the formula tree corresponding to the observed atom are found) is bounded because collapse is essentially the same in Algorithms \ref{alg-recursive-traversal} and \ref{alg-optimised-check}.

\begin{prop}\label{prop-recursive-complexity}
Algorithm \ref{alg-recursive-traversal}, operating on a formula tree $T_{\psi} = (V, E, V^*)$ of a formula $\phi$ with input atom $a$, has worst-case time complexity for the lookup phase $O(|V|)$ where $|V| \le (1 - p^{n+1})/(1 - p)$.  Here, $p = \text{max}\{|\text{l-cl-m}(\psi)|  :  \psi \in \text{cl}(\phi)\}$ and $n$ is the length $k$ of the maximal sequence $(\phi_1, \dots, \phi_k)$ such that $\phi_i \in \text{cl}(\phi)$ and $\phi_i \sqsubset \phi_{i+1}$ for all $1 \le i < k$.
\end{prop}

\begin{proof}
The complexity of lookup in such a recursive traversal algorithm can be determined by finding the number of vertices in the formula tree $T_{\psi}$, since each one must be traversed (in the worst case) to determine the new truth value of the formula $\phi$ given a newly observed symbol $a$.  Hence, it suffices to find a bound for the size of the tree.

Suppose $T_{\psi}$ is the formula tree constructed as in Definition \ref{def-formula-tree} for the formula $\phi$.  Further suppose that the maximum arity of sub-formulas in $\phi$ is $p = \text{max}\{|\text{l-cl-m}(\psi)| : \psi \in \text{cl}(\phi)\}$.  Here, the maximum arity is determined by taking the cardinality of the largest multiset of direct sub-formulas from any formula in the closure of $\phi$.  This is the largest arity found anywhere in $\phi$.  Now, with $p$ and $n$ in mind, consider another formula tree $T_{\psi}' = (V', E', (V')^*)$ which is full with height $n$ and arity $p$, hence has $p^n$ leaves.  The upper bound on the complexity will be the number of vertices in this tree.

The number of vertices is $p^0 + p^1 + p^2 + \dots + p^n$, summing down the rows of the tree $T_{\psi}'$, hence

$$
|V'| = \sum_{k=0}^{n}{p^k} = \frac{1 - p^{n+1}}{1 - p}.
$$

Then, the recursion must visit at most $(1 - p^{n+1})/(1 - p)$ vertices, hence this forms an upper bound (though a loose one).  The result follows.
\end{proof}

\begin{algorithm}[ht]
\begin{algorithmic}[1]

\Procedure{check}{$\phi$, $a$}

\If{$\phi$ \textit{is} conjunction or $\phi$ \textit{is} disjunction}

	\If{$\phi$ \textit{is} conjunction and $\bot \in \text{l-cl}(\phi)$}
		\Return false
	\ElsIf{$\phi$ \text{is} disjunction and $\top \in \text{l-cl}(\phi)$}
		\Return true
	\ElsIf{$\text{l-cl}(\phi) = \{\top\}$}
		\Return true
	\ElsIf{$\text{l-cl}(\phi) = \{\bot\}$}
		\Return false
	\EndIf
	
	\State operands $\gets \text{l-cl-m}(\phi)$
	
	\For{$n \in \{0, \dots, |\text{operands}|\}$}
	
		\If{$\text{operands}_n \not\in \{\top, \bot\}$}
			\If{$\text{operands}_n$ \textit{is} atom}
				\If{$\text{operands}_n = a$}
				
					\State $\text{operands}_n \gets \top$
					
					\If{$\phi$ \textit{is} disjunction}
						\Return true
					\ElsIf{$\phi$ \textit{is} conjunction}
						\State true\_clauses($\phi$)++
						\If{true\_clauses($\phi$) = $|\text{operands}|$}
							\Return true
						\EndIf
					\EndIf
					
				\ElsIf{$\text{operands}_n = \lnot a$}
					\State $\text{operands}_n \gets \bot$
					\If{$\phi$ \textit{is} conjunction}
						\Return false
					\EndIf
				\EndIf
			\Else
				\State sub\_verdict $\gets$ check($\text{operands}_n$, $a$)
				\If{sub\_verdict $= \top$}
					\State $\text{operands}_n \gets \top$
					\If{$\phi$ \textit{is} disjunction}
						\Return true
					\ElsIf{$\phi$ \textit{is} conjunction}
						\State true\_clauses($\phi$)++
						\If{true\_clauses($\phi$) = $|\text{operands}|$}
							\Return true
						\EndIf
					\EndIf
				\ElsIf{sub\_verdict $ = \bot$}
					\State $\text{operands}_n \gets \bot$
					\If{$\phi$ \textit{is} conjunction}
						\Return false
					\EndIf
				\EndIf
			\EndIf
		\EndIf
		
	\EndFor
	
\ElsIf{$\phi$ \textit{is} negation}
	\Return $(\phi = a)$
\EndIf

\EndProcedure
\end{algorithmic}
\algolisting{alg-recursive-traversal}{Recursive Traversal of a Formula Tree}
\end{algorithm}

\subsection{Formula Closure Map Optimisation}

Algorithm \ref{alg-recursive-traversal} requires a traversal of the entire formula tree (which progressively shrinks due to collapses) for every atom observed during runtime.  This is a source of inefficiency, and this section describes an optimisation that has been implemented based on the closure (see Definition \ref{def-formula-closure}) of a formula.  Let $\phi$ be a formula, and let $\text{cl}(\phi)$ be its closure.  Then, let $M$ be a map sending atoms found in $\phi$ (leaves of the tree $T_{\psi}$) to the sub-formulas to which they belong.

\begin{example}\label{example-map-optimisation}
Let $\phi \equiv (p \lor q) \land (p \lor \lnot r)$.  Then $M(p) = \{p \lor q, p \lor \lnot r\}$, $M(r) = \emptyset$ (since the only occurrence in $\phi$ is negative) and $M(\lnot r) = \{p \lor \lnot r\}$.
\end{example}

Algorithm \ref{alg-closure-map-construction} constructs the map that Example \ref{example-map-optimisation} demonstrates.  Using this map, when an atom $a$ is observed, the entire tree $T_{\psi}$ need not be traversed; the vertices of relevance can be found immediately.  It only remains to traverse back up the tree, collapsing vertices as far as is possible; this can be done with a straightforward iteration.

\begin{prop}\label{prop-optimised-complexity}
Algorithm \ref{alg-optimised-check}, operating on a formula tree $T_{\psi} = (V, E, V^*)$ of a formula $\phi$ with input atom $a$, has worst-case time complexity for the lookup phase $O(1)$, and has an amortised complexity of $O((1 - p^{n+1})/(1 - p))$ for the construction of the closure map.  $n$ is, like in Proposition \ref{prop-recursive-complexity}, the bound for the height of the tree $T_{\psi}$.
\end{prop}

\begin{proof}
Lookup of the relevant vertices when a symbol is observed is $O(1)$ by the map constructed by Algorithm \ref{alg-closure-map-construction}.

Construction of the closure map by Algorithm \ref{alg-closure-map-construction} is in

$$
O\Bigg(\frac{1 - p^{n+1}}{1 - p}\Bigg)
$$

since every vertex must be traversed during the recursion.  This process only happens once, hence its complexity is amortised.
\end{proof}

To conclude this presentation of the monitoring algorithms, it should be noted that the size of the formula tree decreases as more information is observed; collapse can never increase the size of the tree.  This means, for a formula $\psi(q_1, \dots, q_n)$, the space complexity for a monitor is $O(|\text{cl}(\psi)|)$ (ie, the number of subformulas, which is the size of the formula tree) during its lifetime.

\begin{algorithm}

\begin{algorithmic}[1]
\Procedure{closure\_map}{$\phi$}

\If{$\phi$ \textit{is} conjunction or $\phi$ \textit{is} disjunction}
	\State operands $\gets \text{l-cl-m}(\phi)$
	\For{$n \in \{0, \dots, |\text{operands}|\}$}
		\If{$\text{operands}_n$ \text{is} atom}
			\State $M(\text{operands}_n) \gets M(\text{operands}_n) \cup \{\phi\}$
		\Else
			\State closure\_map($\text{operands}_n$)
		\EndIf
	\EndFor
\EndIf

\EndProcedure
\end{algorithmic}

\algolisting{alg-closure-map-construction}{Recursive Closure Map Construction}
\end{algorithm}

\begin{algorithm}

\begin{algorithmic}[1]
\Procedure{optimised\_check}{$a$}

\State occurrences $\gets M(a) \cup M(\lnot a)$

\For{$\psi \in \text{occurrences}$}
	\If{$a$ \textit{is} positive}
		\If{$a$ is positive in $\psi$}
			\State Replace $a$ with $\top$ in $\psi$
		\ElsIf{$a$ is negative in $\psi$}
			\State Replace $a$ with $\bot$ in $\psi$
		\EndIf
	\Else
		\If{$a$ is positive in $\psi$}
			\State Replace $a$ with $\bot$ in $\psi$
		\ElsIf{$a$ is negative in $\psi$}
			\State Replace $a$ with $\top$ in $\psi$
		\EndIf
	\EndIf
	
	\State current\_formula $\gets \psi$
	\State collapsed\_value $\gets$ possible truth value of current\_formula
	
	\While{collapsed\_value \textit{is not} undefined}
		\If{current\_formula has a parent formula}
			\State Replace current\_formula with collapsed\_value in parent
			\State current\_formula $\gets$ parent of current\_formula
			\State current\_value $\gets$ possible truth value of current\_formula
		\Else
			\State Set final verdict to collapsed\_value
			\Return collapsed\_value
		\EndIf
	\EndWhile
\EndFor

\Return ?

\EndProcedure
\end{algorithmic}

\algolisting{alg-optimised-check}{Truth value checking with the closure map optimisation constructed in Algorithm \ref{alg-closure-map-construction}.}
\end{algorithm}





\chapter{Instrumentation}\label{chapter-instrumentation}

Instrumentation of a function $f$ for checking a property $\phi_f$ is the process of placing \textit{instruments} (pieces of code that will read some values at runtime) in the code for $f$ such that data taken by them at runtime can be used to decide the truth value of $\phi_f$.  In the language developed in this report, this can be formulated as the following problem:

\begin{displayquote}
Given a symbolic control flow graph $\text{SCFG}_f$ of a function $f$, and a formula $\phi_f$, find the minimal set of vertices and edges in $\text{SCFG}_f$ such that, with appropriate instrumentation of them, information generated by some run modelled by some system $\mathcal{D}$ will be captured by the instruments.  The information captured by the instruments should be the minimal amount required to decide the truth value of $\phi_f$.
\end{displayquote}

\section{A Strategy}\label{section-strategy}

Given the structure of formulas in CTFL given in Definition \ref{def-psi}, the approach taken to instrumentation for a formula $\phi_f$ is to use the combination of the quantification domain(s), and the atoms in the formula (which are all results of predicates).

With this in mind, consider the formula $\phi_f \equiv \forall q \in S_x : \psi(q)$, where $S_x$ is the set of all concrete states changing the value to which the program variable $x$ is mapped.  The strategy used will be to transform $S_x$ into the static context (a set of symbolic states that generate the concrete states at runtime) and use each $sq$ (recall that $sq$ gives the symbolic state that generates $q$) in this set as a \textit{point of interest}, from which the formula $\psi(q)$ can be used to decide which surrounding points in the SCFG must be instrumented.

Denote by $A(\psi(q))$ the set of atoms in the formula $\psi(q)$.  For example, $\psi(q) \equiv q(x) \in [0, 10] \implies d(\text{next}_{\Delta\tau}(q, \nu)) \in [0, 10]$ gives $A(\psi(q)) = \{q(x) \in [0, 10], d(\text{next}_{\Delta\tau}(q, \nu)) \in [0, 10]\}$ (notice that the atoms are predicates applied to bindings from quantification domains).  Notice that $A(\psi(q))$ is independent of $q$.

Further, denote by $s(S_x) = \{\sigma : sq = \sigma \text{ for } q \in S_x\}$ the set of symbolic states $\sigma$ that generate concrete states $q$ (hence, $sq = \sigma$) at runtime.  $s(S_x)$ can be seen as a symbolic support of $S_x$; $s(S_x)$ is the image of the set $S_x$ under $s$ when $s$ is taken as a map.  The set $s(S_x)$ is the set used for instrumentation; each $\sigma \in s(S_x)$ is either a vertex or an edge in $\text{SCFG}_f$.

The next step is to determine, for each $A(\psi(q))$ (where $\psi$ may depend on multiple bind variables but is considered only to depend on one here to simplify notation), the traversal that should be performed on $\text{SCFG}_f$ in order to determine the set of points to which instrumentation should be applied.  For this, let $\alpha \in A(\psi(q))$ be an atom, that is, the value of some predicate applied to some concrete state or transition in a run of $f$ modelled by $\mathcal{D}$.  Hence, $\alpha = P(\tau_i)$ or $\alpha = P(\Delta\tau_i)$ for a predicate $P$ (as an example, let $P(\Delta\tau_i) = d(\text{next}_{\Delta\tau}(\Delta\tau_i, \nu)) \in I$ for a transition $\Delta\tau_i$ and an interval $I$).  Then, the \textit{composition sequence} of an atom $\alpha$ is given in Definition \ref{def-comp-sequence}.

\begin{definition}[Composition Sequence of an atom $\alpha$]\label{def-comp-sequence}
For an atom $\alpha \in A(\psi(q))$, the composition sequence is the sequence $\langle f_1, \dots, f_n, P \rangle$ representing the composition $P \circ f_n \circ \dots \circ f_1$ where:

\begin{itemize}
\item $f_1$ is a map from $q$ to a set of symbolic states or edges.
\item $f_i$ for $1 < i \le n$ are maps from sets of symbolic states or edges in $\text{SCFG}_f$ to other sets of symbolic states or edges.  In particular, for a set of symbolic states or edges $A$, each $f_i$ is such that

$$f_i(A) = \bigcup_{p \in A} f_i(p),$$

ie, $f_i$ maps a set of symbolic states or edges to the union of the images of the members of that set under $f_i$.
\item $P$ is a predicate and has no meaning statically; it requires runtime information to be evaluated.
\item Finally, $(P \circ f_n \circ \dots \circ f_1) = \alpha$.
\end{itemize}

\end{definition}

Based on Definition \ref{def-comp-sequence}, the composition sequences are derived for each $\alpha \in A(\psi(q))$.  It remains to give detail to the $f_i$.  For example, letting $\alpha = d(\text{next}_{\Delta\tau}(q, \nu)) \in [0, 10]$, the composition sequence is $\langle \text{next}_{\Delta\tau}(q, \nu), d(\Delta\tau_i) \in [0, 10] \rangle$ where $f_1 = \text{next}_{\Delta\tau}(q, \nu)$ and $P = d(\Delta\tau_i) \in [0, 10]$, but a way is needed to compute $((d(\Delta\tau_i) \in [0, 10]) \circ (\text{next}_{\Delta\tau}(q, \nu)))(q)$ (ie, to find the next transition satisfying $\nu$ and measure its duration).

To this end, Section \ref{section-static-operators} extends the already defined operators $\text{next}_{\Delta\tau}, \text{next}_{\tau}, \text{future}_{\Delta\tau}$ and $\text{future}_{\tau}$ to make sense statically, since the only definitions given so far are defined in terms of a DDS $\mathcal{D}$.

\section{Future Time Operators in the Static Context}\label{section-static-operators}

Consider the future time operator $\text{next}_{\Delta\tau}(q, \nu)$.  With respect to a model $\mathcal{D}$ of a run of a function $f$, $\text{next}_{\Delta\tau}(q, \nu)$ refers to a single transition; there is a total order on transitions in $\mathcal{D}$, so the notion of \textit{next} yields a unique element.  This is different when considering the meaning of $\text{next}_{\Delta\tau}(q, \nu)$ in the context of an SCFG: there is no guaranteed total ordering on vertices or edges (the only ordering is with respect to paths, and branching means that some elements are incomparable), so $\text{next}_{\Delta\tau}(q, \nu)$ when considered statically must yield a set.

The future time operators in the static context are defined case-by-case.  Before giving their definitions, some terminology is set up:  say an edge $e$ is \textit{before} an edge $e'$ in an SCFG if and only if there is a path $\pi$ starting from the symbolic state at which $e$ originates and ending at the symbolic state at which $e'$ begins.  This notion of a partial order on edges and symbolic states using just symbolic states can be extended to say that a symbolic state $\sigma$ is \textit{before} an edge $e$; an edge $e$ is \textit{before} a symbolic state $\sigma$ and a symbolic state $\sigma$ is \textit{before} a symbolic state $\sigma'$.

Additionally, an edge $e$ in SCFG holds a property $\eta$ (of a form similar Equation \ref{eq-P-delta}, but defined in terms of edges rather than transitions), written $e \models \eta$ (similarly to the notation for transitions generated by a runtime holding properties in Section \ref{section-qds}) if and only if $\eta(e) \equiv \top$.

Now, let $\text{SCFG}_f = \langle V, E, v_s, V_e \rangle$ be a symbolic control flow graph.  Then, for a symbolic state or edge, say $\beta$, in $\text{SCFG}_f$:

\begin{equation}\label{eq-static-next}
\begin{split}
\text{next}_{\Delta\tau}(\beta, \eta) = \{ e \in E : & \; \text{$\beta$ is before $e$},\\
& \; e \models \eta \text{ and}\\
& \; \text{on the path $\pi$ (from $\beta$ to $e$),}\\
& \;\;\; \text{there is no $\beta''$ before $e$ with $\beta'' \models \eta$ and}\\
& \;\;\; \text{where $\beta$ is before $\beta''$.}\}.
\end{split}
\end{equation}

\begin{equation}\label{eq-static-future}
\begin{split}
\text{future}_{\Delta\tau}(\beta, \eta) = \{ e \in E : & \; \text{$\beta$ is before $e$ and } e \models \eta\}.
\end{split}
\end{equation}

Notice that the definition of $\text{future}_{\Delta\tau}$ uses a weaker condition than that of $\text{next}_{\Delta\tau}$: \textit{future} involves all future occurrences, unboundedly.

\begin{remark}
The future time operators in the static context for symbolic states, rather than edges, are defined similarly to Equations \ref{eq-static-next} and \ref{eq-static-future}.
\end{remark}

\begin{example}
Figure \ref{fig-modified-scfg} demonstrates the reachability analysis used in Equations \ref{eq-static-next} and \ref{eq-static-future}.  Figure \ref{subfig-modified-scfg} shows an SCFG.  Consider fixing a symbolic state $q$ to be the symbolic state represented by the vertex $a$.  Then, in this static context, $\text{next}_{\Delta\tau}(q, f)$ (where $f$ is used as a notational shortcut for the predicate that selects calls to $f$) is equal to a set.  Figure \ref{subfig-modified-scfg-highlighted} highlights the vertex $a$ in red and the vertices that are members of $\text{next}_{\Delta\tau}(q, f)$ in blue.  Clearly, $a$ is before every vertex in $\text{next}_{\Delta\tau}(q, f)$.

\begin{figure}
\begin{subfigure}{0.5\linewidth}
\includegraphics[width=\linewidth]{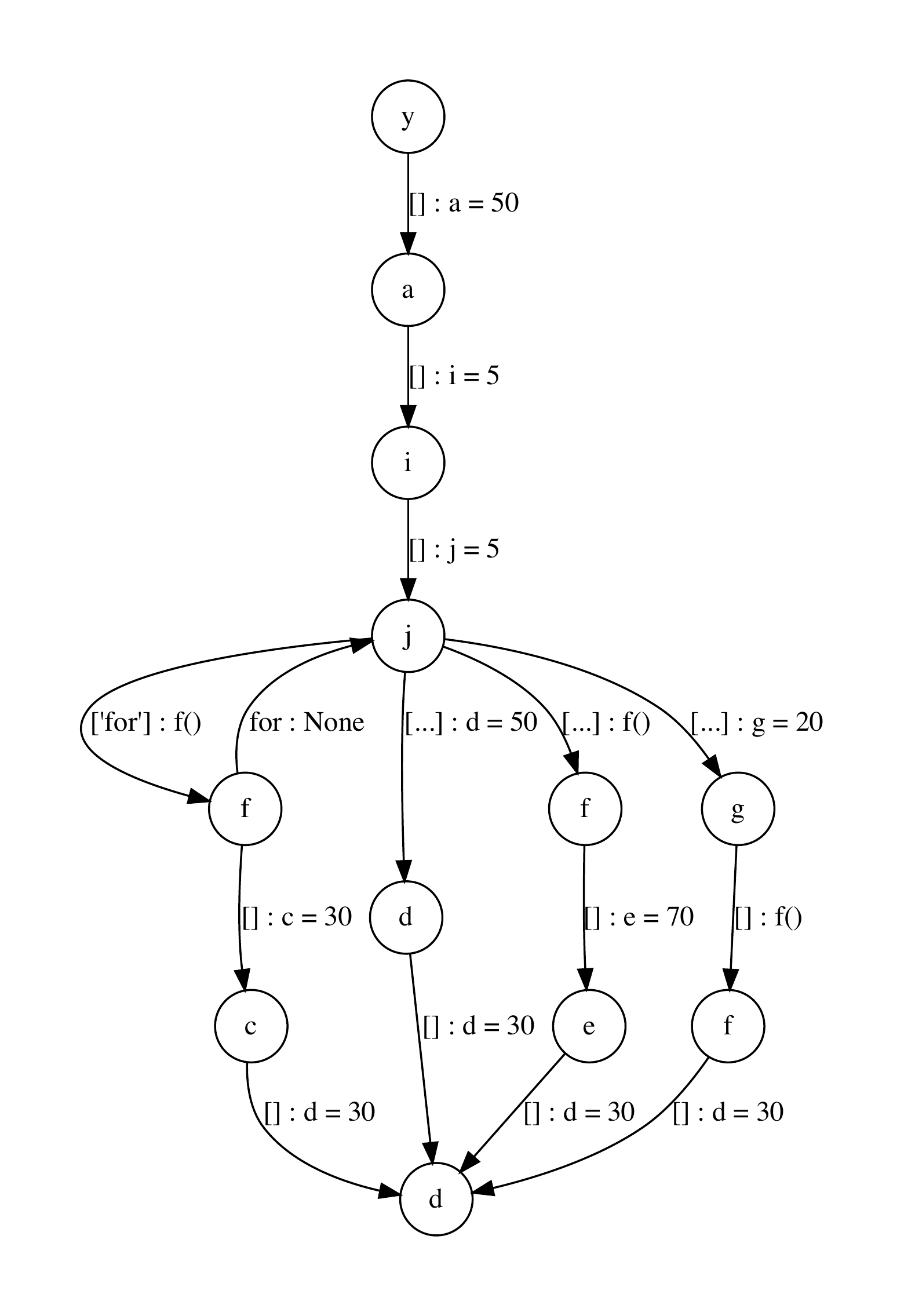}
\caption{\label{subfig-modified-scfg-highlighted}A simplified SCFG.}
\end{subfigure}%
\begin{subfigure}{0.5\linewidth}
\includegraphics[width=\linewidth]{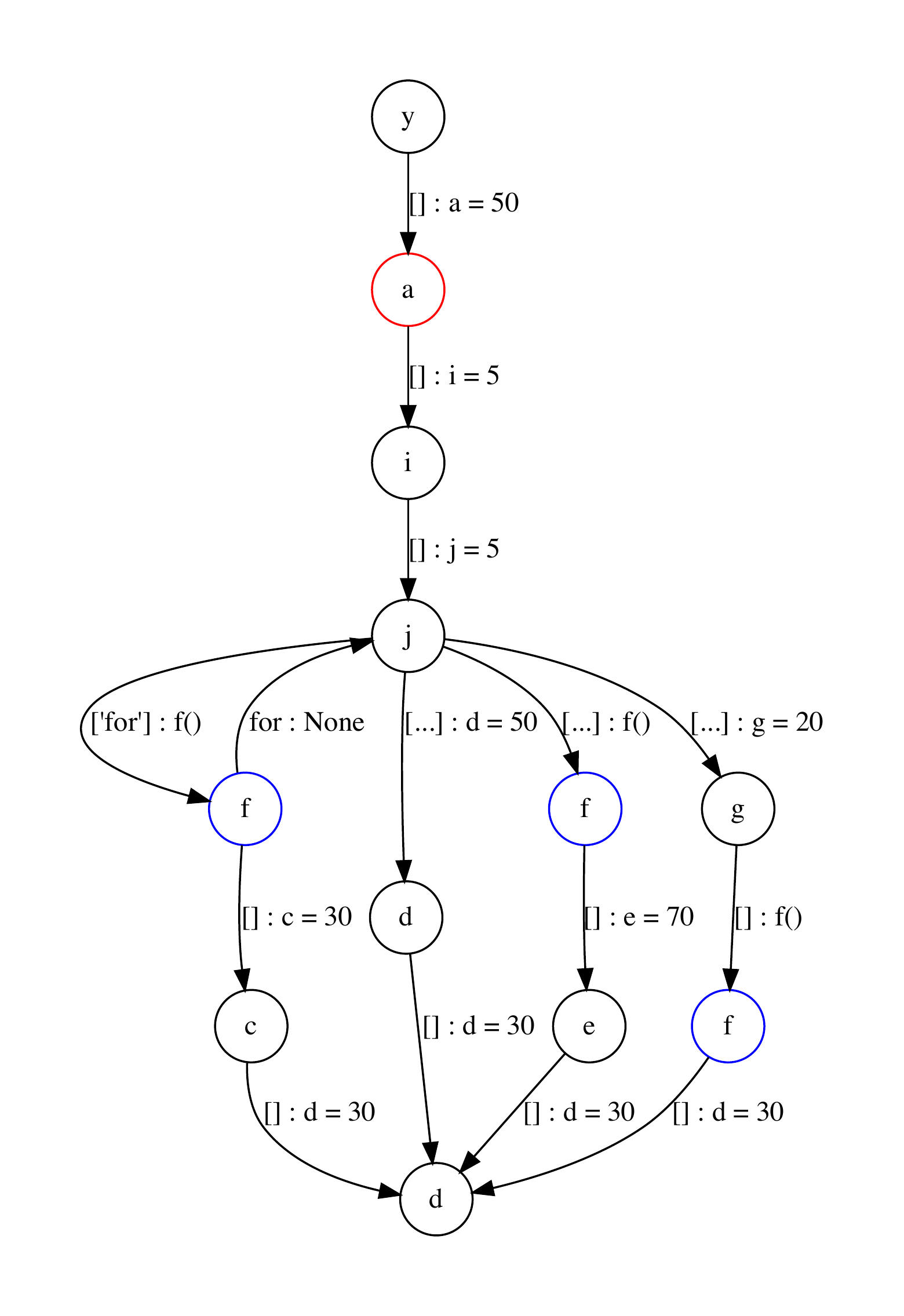}
\caption{\label{subfig-modified-scfg}A simplified SCFG with vertices highlighted according to reachability analysis.}
\end{subfigure}
\caption{\label{fig-modified-scfg}Simplified SCFGs with reachability analysis.}
\end{figure}
\end{example}

The instrumentation strategy for a singly-quantified formula $\phi_f \equiv \forall q \in S_x : \psi(q)$ can, therefore, be summarised as such:

\begin{enumerate}
\item Determine the set of symbolic states $s(S_x)$ of the quantification domain $S_x$.
\item For each $\sigma \in s(S_x)$, derive the set of instrumentation points by combining the set of atoms $A(\psi(q))$ with the composition sequences for each $\alpha \in A(\psi(q))$.
\item Instrument according to the predicate present at the end of each composition sequence.  For example, $P = d(\Delta\tau_i) \in [0, 10]$ means the instrument inserted must measure the duration of the transition $\Delta\tau_i$.
\end{enumerate}

Section \ref{section-multiple-quant} addresses how to deal with multiple quantification; an immediate thought is to form the product of the quantification domains, $S_1 \times \dots \times S_n$, but this is not straightforward since the computation of every $S_j$ ($j > 1$) must depend on a binding $q_{j-1}$ from $S_{j-1}$.

\section{Multiple Quantification}\label{section-multiple-quant}

Suppose a formula $\phi_f$ is given by $\phi_f \equiv \forall q_1 \in S_1 , \dots , \forall q_n \in S_n : \psi(q_1, \dots, q_n)$.  The major source of complication here is that for each $1 < i \le n$, $S_i$ must depend on a binding from some $q_j \in S_j$ for some $j < i$, hence the binding cannot be generated by simply taking the product $S_1 \times \dots \times S_n$ of $n$ sets; there is dependency between them.

It follows that one must explicitly recompute the quantification domains given a binding from another set on which they depend.  Take as an example $\phi_f \equiv \forall q \in S_x, \forall t \in \text{future}_{\Delta\tau}(q, \nu) : \psi(q, t)$.  In this case, to obtain the binding $t$, one must compute the set $\text{future}_{\Delta\tau}(q, \nu)$, which itself requires a binding $q$ from $S_x$.

Since the current problem of interest is instrumentation, the computation of bindings from multiple quantifiers must be considered in the static context.  It is clear that $S_x$ must be computed first since this has no dependence on other quantification domains.  This can be computed statically by finding the symbolic states in the SCFG.  Now, consider $\text{future}_{\Delta\tau}(\beta, \nu)$ where $q$ is replaced by $\beta$ to emphasise the static context.  Using the definition of \textit{future} in Equation \ref{eq-static-future}, computing this set is straightforward and requires reachability analysis on the SCFG.  In the context of instrumentation, this approach generates a set of pairs of instrumentation points which, at runtime, generates a possibly larger set of pairs (since single symbolic states/edges can correspond to multiple concrete states/transitions at runtime) on which the verification is actually performed.

\subsection{Instrumenting for Multiple Quantification}

The notion of static pairs can be extended to static bindings for any $n$-quantified formula.  The notion of a \textit{Binding Space} is now given in Definition \ref{def-binding-space}.

\begin{definition}[Binding Space]\label{def-binding-space}
Given $\text{SCFG}_f$ of a function $f$ and a formula $\phi_f \equiv \forall q_1 \in S_1, \dots, \forall q_n \in S_n : \psi(q_1, \dots, q_n)$, the binding space $\mathcal{B}$ of $\phi_f$ with respect to the quantification sequence $\forall q_1 \in S_1, \dots, \forall q_n \in S_n$ is set of tuples $(sq_1, \dots, sq_n)$ such that, if some $S_j$ ($j > 1$) depends on $S_{j-1}$, then $sq_j$ is a member of the set $s(S_j)$ computed using the binding $sq_{j-1} \in s(S_{j-1})$.
\end{definition}

Instrumentation is performed for multiply-quantified formulas by computing the binding space $\mathcal{B}$ for an SCFG entirely and, for each binding $B \in \mathcal{B}$, computing the set of instrumentation points $I(\psi, B)$ by applying the reachability analysis derived from the composition sequences of atoms in $A(\psi)$ to the appropriate components of the binding $B \in \mathcal{B}$.

Algorithm \ref{alg-construct-binding-space} is used to statically construct the binding space $\mathcal{B}$ from an SCFG using the quantifier sequence of the formula $\phi_f$ and a \textit{reachability map}, which is now defined.  A reachability map $R : V \to (V \cup E)$ for some $\text{SCFG}_f = \langle V, E, v_s, V_e \rangle$ is a map sending vertices in $\text{SCFG}_f$ to sets of reachable vertices and edges.  Its construction is trivial; the vertices in $V$ are iterated over and a depth-first search is performed.  It also makes sense to consider a similar map $E \to (V \cup E)$, and indeed such a map is required.

Given some $\text{SCFG}_f$, a formula $\phi_f$ and a reachability map $R$, the idea of Algorithm \ref{alg-construct-binding-space} is as follows.  Given a binding $B = (sq_1, \dots, sq_k)$ for $k < n$ for $\phi_f$ containing a sequence of $n$ quantifiers (hence, the binding given may be partial), then, recursively:

\begin{itemize}
	\item If $B$ has length 0, then the computation is begun by computing the independent quantification domain $S_1$, since this necessarily does not require a binding from a previous domain for its computation.
	
	Computation of $S_1$ is straightforward; either the vertex set or edge set is filtered depending on the predicate that specifies $S_1$.
	
	For each of the (static) elements $sq_1$ of $S_1$, recurse on that element with the partial binding $(sq_1)$.
	
	\item If $B$ has length $k < n$, hence is still a partial binding, let the quantifier for which the next part of the binding is computed be that at position $k+1$ in the quantifier sequence of $\phi_f$.  Then, to compute the next part of the binding, the quantifier on which it depends is found so that the value at that position in the binding can be used.
	
	Given the binding $(sq_1, \dots, sq_k)$ for $1 \le k < n$, this step aims to compute the bindings $(sq_1, \dots, sq_k, sq_{k+1})$ for $sq_{k+1} \in S_{k+1}$, where $S_{k+1}$ is computed based on some binding $sq_j$ from $S_j$ for $1 \le j < k+1$.  Hence, when the quantifier is found, the corresponding value in the current binding is taken and used in the computation of $S_{k+1}$.  The computation performed depends on the type of $S_{k+1}$.  Take, for example,
	
	$$
	\phi_f \equiv \forall q \in S_x, \forall t \in \text{future}_{\Delta\tau}(q, \nu) : \psi(q, t).
	$$
	
	Then $S_2(q) = \text{future}_{\Delta\tau}(q, \nu)$, and so computation of $\text{future}_{\Delta\tau}$ is performed using the static-context definition (Equation \ref{eq-static-future}, and similarly for $\text{future}_{\tau}$).
	
	At this point, the reachability map $R$ is used; restricting attention to symbolic states of the SCFG, for some $sq_i$, $\text{future}_{\Delta\tau}(sq_i, \nu) = \{e : e \in R(sq_i) \cap E \text{ and } e \models \nu\}$ (the set of edges satisfying $\nu$ and reachable from $sq_i$).  Each binding in the computed set is then recursed on.
	
	\item If $B$ has length $k$, hence is a complete binding, add it to the binding space $\mathcal{B}$.
\end{itemize}

\begin{algorithm}

\begin{algorithmic}[1]
\Procedure{compute\_binding\_space}{($\forall q_1 \in S_1, \dots, \forall q_n \in S_n$), $\text{SCFG}_f$, $R$, $B$}

\If{$B$ is empty}

\State $S_1$ $\gets$ filter $V$ or $E$ according to the predicate in $S_1$
\State $B'$ $\gets$ empty list
\For{$B'' \in S_1$}
	\State $B' \gets B' \cup \text{compute\_binding\_space}(\forall..., \text{SCFG}_f, R, (B''))$
\EndFor

\Return B'
	
\ElsIf{length of $B$ is $k < n$}

\State next\_quantifier \_ $ \gets (\forall q_{k+1} \in S_{k+1})$
\State required\_quantifier $\gets (\forall q_j \in S_j)$ such that $S_{k+1}$ depends on $S_j$
\State current\_binding $\gets sq_j$ from $B$
\State $S_{k+1} \gets$ use $R$ to compute $S_{k+1}$
\State $B' \gets$ empty list
\For{$B'' \in S_{k+1}$}
	\State $B' \gets B' \cup \text{compute\_binding\_space}(\forall..., \text{SCFG}_f, R, B + B''))$
\EndFor

\Else

	\Return $\{B\}$
\EndIf

\EndProcedure
\end{algorithmic}

\algolisting{alg-construct-binding-space}{Recursive Construction of a Binding Space}
\end{algorithm}

Now, it remains to define how this process relates to monitoring (which Chapter \ref{chapter-cftl-monitoring} describes in detail).  For example, $\phi_f \equiv \forall q \in S : \psi(q)$ is a straightforward case and requires $|S|$ monitors; $S$ depends on no other quantification domain, so its only dependence is on the code of $f$ (though Section \ref{section-monitor-res} will address the problem of when to instantiate new monitors, update existing ones and do nothing).  A more complex example is $\phi_f \equiv \forall q_1 \in S_1, \dots, \forall q_n \in S_n : \psi(q_1, \dots, q_n)$, where the number of monitors to be instantiated is not straightforward.  As mentioned before, one cannot simply take the product of the quantification domains and say that $\Pi_{i=1}^{n}|S_i|$ monitors are required.  In fact, this makes no sense since $S_j$ ($j > 1$) must require a binding from $S_{j-1}$ for their computation (hence, they possibly change, along with their cardinalities, when the binding on which they depend changes).  This collection of problems can be seen as being part of one single problem: monitor resolution.  While this problem is solved by the monitoring algorithm described in Chapter \ref{chapter-cftl-monitoring}, the solution requires some theory to be set up in the static context.  Hence, the remainder of this section considers the problem of:

Given a point in the program that has been instrumented, to which monitors is the data sent?  This question also considers whether new monitors must be instantiated.

\section{Monitor Resolution}\label{section-monitor-res}

Suppose the binding $B = (\alpha_1, \dots, \alpha_n)$ has been computed statically, that is, $\alpha_1, \dots, \alpha_n$ are either symbolic states or edges of $\text{SCFG}_f$.  Suppose further that the formula in question is $\phi_f \equiv \forall q_1 \in S_1, \dots, \forall q_n \in S_n : \psi(q_1, \dots, q_n)$.  The problem of monitor resolution is, given a set of points that must be instrumented when given both the binding $B$ and the subformula $\psi(q_1, \dots, q_n)$, for each of those points, what should be the behaviour regarding monitor instantiation or update when data is received from the instrumented points at runtime?

To this end, recall that $I(\psi, B)$ denotes the set of instrumentation points required to decide $\psi(q_1, \dots, q_n)$ given the binding $B$ (this can be computed using the techniques discussed so far).  Then, for $\text{SCFG}_f = \langle V, E, v_s, V_e \rangle$,  $I(\psi, B) \subset V \cup E$, that is, the set of instrumentation points can be said to form a subgraph\footnote{Technically one then has to form edges in this subgraph by allowing an edge to exist iff there is a path between vertices/edges in $\text{SCFG}_f$.} of $\text{SCFG}_f$.

Since $I(\psi, B) \subset V \cup E$, one can form a natural ordering on $I(\psi, B)$ with respect to paths in $\text{SCFG}_f$.  This ordering is the same as that described in Section \ref{section-static-operators}: for $u, v \in I(\psi, B)$, $u$ is \textit{before} $v$ in $I(\psi, B)$ $\iff$ $u$ is \textit{before} $v$ in $\text{SCFG}_f$.  The monitoring algorithm that Chapter \ref{chapter-cftl-monitoring} presents uses this notion of a partial order on SCFGs heavily.

\chapter{Monitoring for CFTL Formulas}\label{chapter-cftl-monitoring}

This chapter of the report will discuss the approach taken to check the truth value of formulas $\phi_f$ in CFTL at runtime, assuming the instrumentation problem described in Chapter \ref{chapter-instrumentation} has been solved.

Consider again the partial order on SCFGs discussed in Chapter \ref{chapter-instrumentation} and recall that $I(\psi, B)$ is the set of instrumentation points computed using the atoms in $A(\psi)$ and the binding $B$.  Initially considering formulas to be those that are single-quantified, this partial ordering lifts to $I(\psi, B)$ to allow one to write the first rule about monitor resolution for singly-quantified formulas:

\begin{displayquote}
An element $v \in I(\psi, B)$ can trigger instantiation of a new monitor iff it is minimal with respect to the \textit{before} ordering on $I(\psi, B)$.
\end{displayquote}

The intuition behind this is that minimal elements with respect to \textit{before} will be observed first, so should be able to trigger instantiation of a new monitor.  Consider the formula

$$\phi_f \equiv \forall q \in S_x : q(x) = 10 \implies \text{next}_{\tau}(q, y)(y) = 15,$$

where $\text{next}_{\tau}(q, y)$ is an abuse of notation to mean ``the next state that changes $y$''.  Then,

$$I(q(x) = 10 \implies \text{next}_{\tau}(q, y)(y) = 15, B)$$

for some static binding $B$ is a set with a natural partial ordering that insists that $q(x) = 10$ should be observed before $\text{next}_{\tau}(q, y)(y) = 15$; observing $\text{next}_{\tau}(q, y)(y) = 15$ should not result in instantiation of a new monitor since $q(x) = 10$ may never be observed (especially if the property is not being verified inside a loop).  This is the intuition behind using the partial ordering to determine the atoms whose observation may result in instantiation of a new monitor.

\section{Monitor Update vs Doing Nothing}

It remains to define what should happen when an atom is observed that is not the minimal in a partial ordering, but with which something can still be done with respect to existing monitors.  With this in mind, consider the code:

\begin{lstlisting}[numbers=left]
function f()
  a = 10
  for i = 0 to 9
    g(i)
\end{lstlisting}

and consider verification with respect to the property

$$\phi_f \equiv \forall q \in S_a : q(a) = 10 \implies d(\text{next}_{\Delta\tau}(q, g)) \in [0, 10].$$

Observing $a = 10$ will make the atom $q(a) = 10$ true, then if the call to $g$ takes time within $[0, 10]$ units, the property is satisfied (there is only one instruction that changes $a$).  But then the loop carries on for another 9 iterations and the instrument placed at the call to $g$ (using the current method) records another 9 calls to $g$.  The question is then: what is done with these remaining occurrences, if the monitor instantiated to deal with \lstinline{a = 10} has reached a verdict after the first iteration of the loop?

This is a case in which nothing is done; the formula asks for the \textit{next} occurrence of a call to $g$, and this was taken into account, leading the monitor to reach a verdict.  Naturally, if there is no further occurrence of \lstinline{a = 10}, no further monitors should be instantiated.

There are, therefore, two cases addressed here:

\begin{itemize}
	\item The first call to $g$ is in a context where there is a monitor (instantiated by \lstinline{a = 10}, since that instruction is a minimal element with respect to the partial order on $I(q(a) = 10 \implies d(\text{next}_{\Delta\tau}(q, g)) \in [0, 10], B)$ and so can trigger instantiation of a new monitor) which has not yet received information about a call to $g$.  This means observation of the call to $g$ is taken into account by the existing monitor.  This is a \textit{monitor update}.
	
	\item The subsequent calls to $g$ are in a context where there is only one monitor and it has reached a verdict.  They must, therefore, \textit{do nothing} for verification.
\end{itemize}

Consider now a modified version of the code (where, also, $i$ is added to $a$ to be sure that its value changes and the execution of the instruction in every iteration is included in the quantification domain):

\begin{lstlisting}[numbers=left]
function f()
  for i = 0 to 9
    a = 10 + i
    g(i)
\end{lstlisting}

Here, the monitoring story is different.  For each loop iteration, an instruction is executed which is minimal with respect to the partial order of the set of instrumentation points.  It is therefore allowed to instantiate a new monitor.  This leads to a monitor being instantiated for each iteration of the loop; there are 10 instances of verification (since $S_x$ in the formula being monitored refers to concrete states generated at runtime), and so each subsequent call to $g$ also contributes towards reaching a monitor verdict (each call contributes to the verdict of a different monitor).

\section{Monitoring Multiply-Quantified Formulas}

Consider again the code

\begin{lstlisting}[numbers=left]
function f()
  a = 10
  for i = 0 to 9
    g(i)
\end{lstlisting}

with, instead, the property

\begin{equation}\label{equation-monitoring-pnf}
\phi_f \equiv \forall q \in S_a , \forall t \in \text{future}_{\Delta\tau}(q, g): q(a) = 10 \implies d(t) \in [0, 10].
\end{equation}

Notice that this can be rewritten out of prenex normal form:

$$\phi_f \equiv \forall q \in S_a : q(a) = 10 \implies \forall t \in \text{future}_{\Delta\tau}(q, g): d(t) \in [0, 10].$$

So this monitor expresses the property ``If $a$ is changed to equal 10, then all future calls to $g$ should take time within $[0, 10]$ units''.  By using Equation \ref{equation-monitoring-pnf}, monitoring can be performed by first computing the pair $(q, t)$ for some concrete state $q$ and some transition $t$ (so in the context of runtime) and instantiating a monitor specifically for that pair.  Now, let $S_i^* = S_i$ such that $S_i$ has the greatest cardinality of all the versions of $S_i$ given the binding from the previous quantification domain on which it might depend.  Then, monitoring the formula with multiple quantification

$$\phi_f \equiv \forall q_1 \in S_1, \dots, \forall q_n \in S_n : \psi(q_1, \dots, q_n)$$

requires $O(\Pi_{i=1}^{n}|S_i^*|)$ monitors.  This bound, however, can only be made precise once the runtime has finished; some monitors may only be instantiated because of loops so, just from this, the size of the quantification domains cannot be statically determined (since one cannot guarantee that a loop's number of iterations can be statically bounded).

\subsection{Multiple Quantification as a Product}

Consider again the multiply-quantified formula in Equation \ref{equation-monitoring-pnf}.  The approach taken in this report to monitoring for this formula (and that ties in well with the instrumentation strategy for multiply-quantified formulas given in Chapter \ref{chapter-instrumentation}) is to take the sequence of quantifiers and the bindings they generate ($(q, t)$ in this case) and consider these bindings as members of some space with the same dimension as the number of quantifiers.  This space draws similarities with the Binding Space (Definition \ref{def-binding-space}); the binding space can be seen as a \textit{symbolic support} for this space of tuples $(q, t)$ (a familiar concept in the relationship between the static and runtime contexts).  For example, in the space of tuples $(q, t)$ (that is, the space of pairs of concrete states and transitions), each $(q, t)$ is generated by some static tuple in the binding space generated during instrumentation.

This space of tuples of concrete states/transitions is essentially the product of the quantification domains $S_a$ and $\text{future}_{\Delta\tau}(q, g)$, but contains more structure since computation of $\text{future}_{\Delta\tau}(q, g)$ requires a binding from $S_a$.  This means the standard definition of the Cartesian Product (for sets $A$ and $B$ where the contents of each do not depend on the other set, $A \times B = \{(a, b) : a \in A, b \in B\}$) does not work.  For example, when one constructs a product $A \times B$, one approach could be to fix some $a \in A$ and generate pairs by iterating through the elements $b \in B$.  Fixing the next $a \in A$, the same elements of $B$ would be iterated over.  In the context of quantification domains and their products, the second $a \in A$ would require iteration over a different version of $B$ since $B$ now depends on an element of $A$.

Nevertheless, once such a product space is computed, it may be traversed.  Hence, monitoring the formula

$$
\phi_f \equiv \forall q \in S_a , \forall t \in \text{future}_{\Delta\tau}(q, g): q(a) = 10 \implies d(t) \in [0, 10]
$$

turns into monitoring the formula

$$
\phi_f \equiv \forall (q, t) \in P: q(a) = 10 \implies d(t) \in [0, 10],
$$

where $P$ is the product space of the quantification domains $S_a$ and $\text{future}_{\Delta\tau}(q, g)$.  From this, a monitoring algorithm begins to take shape; the space $P$ can be seen as the space generated at runtime by the binding space computed statically.  Hence, going by the same rationale as for monitoring singly-quantified formulas, each tuple in the product space of the quantification domains can correspond to a separate monitor and, for $\phi_f$ to hold, all of these monitors must collapse to $\top$.  The complication appears when one attempts to define the rules for when monitors may be instantiated.

Consider the rule for the single-quantified case, applied to multiple-quantification:

\begin{displayquote}
An element $v \in I(\psi, B)$ can trigger instantiation of a new monitor iff it is minimal with respect to the \textit{before} ordering on $I(\psi, B)$.
\end{displayquote}

Consider, also, the simple program

\begin{lstlisting}
a = 20
for i = 0 to 5
  g(i)
\end{lstlisting}

and the property

$$
\phi_f \equiv \forall q \in S_a, \forall t \in \text{future}_{\Delta\tau}(q, g) : q(a) = 20 \implies d(t) \in [0, 10].
$$

The monitoring story is then:

\begin{enumerate}
	\item \lstinline{a = 20} is observed, and is a minimal element in $I(\psi, B)$, so triggers the instantiation of a monitor whose verdict is ? because it has not observed enough to collapse $q(a) = 20 \implies d(t) \in [0, 10]$ to any truth value.
	
	\item \lstinline{g(0)} is observed, is not a minimal element in $I(\psi, B)$, but there is an existing monitor that has not yet observed this data.  The observed data is therefore sent to that monitor, which collapses to a truth value.
	
	\item \lstinline{g(1)} is observed.  The property asserts that \textit{every future call to $g$} must satisfy this property, and yet the instrumentation point at $g$ is not minimal in $I(\psi, B)$, so cannot instantiate a new monitor; there is a problem.
\end{enumerate}

From this it is clear that, for multiple quantification, the minimality condition does not suffice.  Instead, one must partition the instrumentation set $I(\psi, B)$ into sets of instrumentation points derived from each bind variable.  Therefore, with respect to the quantification sequence $\forall q_1 \in S_1, \dots, \forall q_n \in S_n$, a partition $I_1, \dots, I_k \subset I(\psi, B)$ is a family of sets such that:

$$
\forall I_i \; \exists sq_j \; \forall p \in I_i : p \text{ is derived from $sq_j$},
$$

and the usual conditions for partitions of sets hold:

\begin{itemize}
\item The $I_i$ are pairwise disjoint, that is, $\forall 1 \le i, j \le k$ with $i \neq j$, $I_i \cap I_j = \emptyset$, and
\item $\bigcup_{i=1}^{k}I_i = I(\psi, B)$.
\end{itemize}

Now, given that a set exists in the family for each $S_i$, consider the partial order on $I(\psi, B)$ again, but restricted to each $I_i$ in the partition family.  Then a new condition may be written:

\begin{displayquote}
An element $v \in I(\psi, B)$ can trigger instantiation of a new monitor iff it is minimal with respect to the \textit{before} ordering on $I_j \subset I(\psi, B)$.
\end{displayquote}

The partial ordering induced on $I_j$ by that on $I(\psi, B)$ is the expected one.  Before attempting to monitor the code above again, it is necessary to define one final structure: a map from bindings in the binding space (hence, static bindings) to sets of \textit{configurations}, that is, other maps that describe the truth value of every atom in a monitor when it is finally collapsed to a truth value.  The purpose of such a structure is as follows:

Consider a binding at runtime $(q_1, \dots, q_n)$ in the product space of quantification domains, for which a monitor has been instantiated and has reached a verdict (so collapsed to a truth value).  Then, consider a new binding $(q_1, \dots, q'_n)$ where $q'_n$ is the next element from the final quantification domain $S_n$.  Assuming that monitors are deleted when they are collapsed to a truth value, to instantiate a monitor for the binding $(q_1, \dots, q'_n)$ then requires data that has been observed and thrown away, namely the variables $q_1, \dots, q_{n-1}$ from the binding.  Hence, in order to instantiate a monitor for $(q_1, \dots, q'_n)$, knowledge of $q_1, \dots, q_{n-1}$ (and any implicit values derived for future-time operators in the formula) is required.  A solution to this is to set up a map of configurations.

Let $\textsc{Conf}(M_\psi) : A(\psi) \to \{\top, \bot, ?\}$ be a configuration that, for a monitor $M_\psi$, gives the observed truth values of the atoms in the formula $\psi$.  A map of configurations is a map

$$M : \mathcal{B} \to \{\textsc{Conf}(M) : M \text{ is a collapsed monitor for $\psi$}\}$$

that sends static bindings to the set of configurations reached by monitors (for that binding) that were collapsed.  This acts as storage for monitor states; whenever a monitor is collapsed to a truth value during the monitoring process, its final state is stored in this map.  Using this map, when an instrumentation point is executed during runtime that is minimal in the set in the partition to which it belongs, a new monitor is instantiated for each previous monitor configuration found in this map for the relevant binding.

Using these new mechanisms, monitoring is attempted again on the code

\begin{lstlisting}
a = 20
for i = 0 to 5
  g(i)
\end{lstlisting}

with respect to the property

$$
\phi_f \equiv \forall q \in S_a, \forall t \in \text{future}_{\Delta\tau}(q, g) : q(a) = 20 \implies d(t) \in [0, 10].
$$

The new monitoring story is:

\begin{enumerate}
	\item \lstinline{a = 20} is observed, and triggers a new monitor since it is minimal in the set of instrumentation points derived from $q \in S_a$.
	\item \lstinline{g(0)} is observed and takes less than 10 units of time, and a monitor already exists, hence is updated and is collapsed to a truth value.  There are no more monitors left, but the map $M$ now has
	
	$$M((\text{\lstinline{a = 20}, \lstinline{g(i)}})) = \{[q(a) = 20 \mapsto \top, d(t) \in [0, 10] \mapsto \top]\}.$$
	
	\item \lstinline{g(1)} is observed and takes less than 10 units of time, but no monitors exist.  However, \lstinline{g(i)} is minimal in the set of instrumentation points derived from $t \in \text{future}_{\Delta\tau}(q, g)$, so the map $M$ is used; a new monitor is instantiated for every configuration found.  Given that the static binding to which \lstinline{g(1)} belongs is $(\lstinline{a = 20}, \lstinline{g(i)})$, this is given as input to $M$ and the result is the set
	
	$$\{[q(a) = 20 \mapsto \top, d(t) \in [0, 10] \mapsto \top]\}.$$
	
	Now, a monitor is instantiated for every member using every value except that which has just been observed (since only the \textit{old} values are required).  Hence, $q(a) = 20 \mapsto \top$ is used, but $d(t) \in [0, 10] \mapsto \top$ is not since a new value for this has just been observed.
	
	\item This process repeats while more calls to \lstinline{g(i)} are observed.
\end{enumerate}

From this, one can see that the idea of storing old configurations of monitors is a way to remember the results of monitors for different bindings in the product of the quantification domains.  For example, considering again the example $(q_1, \dots, q_n)$ with the second binding $(q_1, \dots, q'_n)$; the use of the map $M$ allows some backtracking in the form of recovering monitor states as if the $q_1, \dots, q_{n-1}$ were to be observed again.

\subsection{A Monitoring Algorithm}

Monitoring multiply-quantified formulas (of which singly-quantified formulas are a special case) is performed in this work by application of a set of rules, all of which have now been developed.  They are presented here for clarity.  Suppose some data is received from an instrument, and that the binding to which this instrument belongs has been computed (this is straightforward; one simply maintains a map when placing the instruments).  Then monitoring is performed by applying the following:

\begin{itemize}
\item If the data observed is from an instrument that is minimal with respect to the partial order on the partition set to which it belongs,
\begin{itemize}
	\item If there are no existing monitors for this binding, and the bind variable to which this data corresponds is the first one, instantiate a new monitor.
	\item If there are no existing monitors for this binding, and the bind variable to which this data corresponds is not the first one, instantiate new monitors for every configuration found in the map $M$.  To select from $M$, suppose the bind variable to which this data corresponds is the $k^{\text{th}}$.  Then, find all bindings for which configurations are stored whose first $k-1$ bind variables match the first $k-1$ bind variables of the current binding.
	
	For each of these configurations, instantiate a new monitor and update the monitor with the data from that configuration, leaving out the atom that has just been observed.  Once this update is complete, update the monitor with the atom based on the newly observed data.
	
	\item If there are existing monitors for this binding, and the bind variable to which this data corresponds is the first one, instantiate a new monitor.
	
	\item If there are existing monitors for this binding, and the bind variable to which this data corresponds is not the first one, update existing monitor states.
\end{itemize}

\item If the data observed is from an instrument that is not minimal with respect to the partial order on the partition set to which it belongs, update existing monitors.

\end{itemize}

\chapter{Performance Evaluation}\label{chapter-performance}

In this chapter, the performance of the verification system that has been developed for CFTL will be analysed in the online monitoring setting.

The verification system developed is asynchronous in that the process of checking for satisfaction of a property does not block the program under scrutiny.  It should also be noted that, currently, the verification is done only for the purpose of analysis; no automatic adaptation is performed on the verified program's trajectory to avoid future violations.  In addition, since the instrumentation is performed statically, its performance will not be analysed.  This is because it is seen as amortised, and so the overhead that is of interest in this chapter is not affected by it.

Finally, the verification system analysed in this chapter is written in Python, for verification of Python programs.  This is because the CMS Collaboration favours the use of Python, hence a lot of the systems that will likely be verified will be written in it.  Further, Python provides many introspection features that make instrumentation more straightforward.

\section{Experimental Setup}

The experiments in this chapter were performed by setting up two threads; one program thread (in which the instrumented version of an input program is run) and one monitoring thread (to which instruments in the program thread send their data for verification).

The plots were generated by modifying the verification code to add \textit{timing points} to an SQLite file, which was then queried for plotting.

The verification system presented for CFTL takes as input the program, and the property for which the program is checked.  It instruments the program according to the theory in Chapter \ref{chapter-instrumentation} and then monitors it at runtime according to the theory in Chapter \ref{chapter-cftl-monitoring}.  The details of implementation will be described in a future paper.

The programs verified in this chapter are artificial examples written to mimic patterns found in CMS' Conditions release service \cite{Dawes}.  By verifying these programs, it is demonstrated that CFTL can indeed be used to write the properties for which it was designed.

The program examples considered cover the following cases:

\begin{itemize}
	\item Performing a database operation on a database specified by some variable, and then closing the database.  The database operations are subject to time constraints, but only if the database given is of a certain type.
	
	\item Performing a sequence of database operations on a database specified.  Each operation in the sequence is subject to a constraint.
	
	\item Checking of a lock that controls data upload to single source.  If the credentials of the uploader are authenticated, then acquiring a new lock on the shared resource constitutes a series of database queries, each of which must be subject to time constraints.
\end{itemize}

\subsection{Verifying Representative Programs}

Each program representative of the cases listed above will now be presented, along with analysis of its verification results.  The analyses will consist of presentation of plots derived from the verification tool; these plots give good insight into how the implemented verification tool works.

\subsubsection{Database Operation followed by Closure}

\begin{figure}[ht]
\begin{lstlisting}
database = 1
database_operation(database)
close_connection(database)
\end{lstlisting}
\caption{\label{fig-code-example-1}Code that sets a database type, calls a function \lstinline{database_operation}, and then closes the database connection with \lstinline{close_connection}.}
\end{figure}

Figure \ref{fig-code-example-1} is a simple code snippet, but serves to demonstrate verification of an important property:

\begin{equation}\label{eq-code-example-1-property}
\begin{split}
\phi \equiv \forall q \in S_d : q(\text{database}) = 1 \implies & \; (d(\text{next}_{\Delta\tau}(q, \text{database\_operation})) \in [0, 2]\\
& \; \land d(\text{next}_{\Delta\tau}(q, \text{close\_connection})) \in [0, 1])
\end{split}
\end{equation}

\begin{figure}[ht]
\centering
\includegraphics[width=0.8\linewidth]{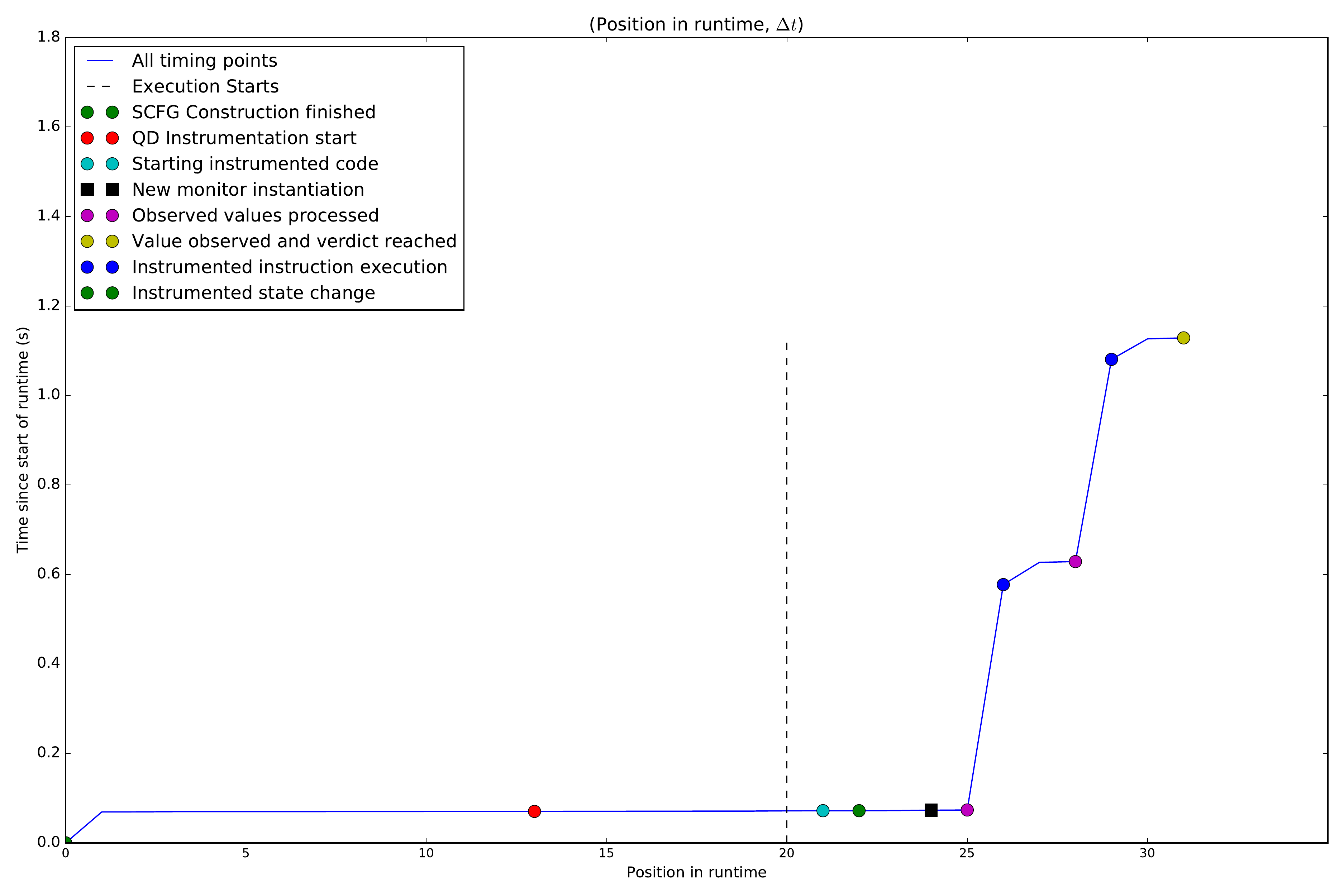}
\caption{\label{fig-code-example-1-plot}Plot showing the progress of verification, with key points in the process followed by the verification tool highlighted.}
\end{figure}

Figure \ref{fig-code-example-1-plot} shows the progress of the verification tool when run on the code in Figure \ref{fig-code-example-1} with the property in Equation \ref{eq-code-example-1-property}.  The vertical dotted line denotes the end of static instrumentation and the beginning of execution of the instrumented code.  The other markings on the plot are as follows:

\begin{itemize}
\item The green circle at the origin denotes the end of the construction of the input program's SCFG.

\item Red circles denote instrumentation of the code with respect to an element of the static binding space.  Since the property in Equation \ref{eq-code-example-1-property} is singly-quantified, the binding space is one-dimensional.  Further, there is only one red circle since there is only one change of the variable \lstinline{database} found by static analysis.

\item The blue circle entitled \textit{Starting instrumented code} denotes the start of the instrumented code in the program thread.

\item The green circle between positions 20 and 25 in the runtime denotes execution of an instrumented state change, ie, a statement has been executed that changed a state that is of interest to the property in Equation \ref{eq-code-example-1-property}.

\item The black square denotes instantiation of a new monitor.  In this case, the state change triggered previously corresponds to an instrumentation point that is minimal in its instrumentation point set, so can trigger instantiation a new monitor; the black square denotes such an instantiation.

\item The purple circle denotes that the observed value (from the state change) has been processed by the monitor and its state has changed (due to observation of an atom).

\item The blue circle, after the jump (which corresponds to a call of the function \lstinline{database_operation}), denotes receipt of data from the instrument around the function call.

\item The next purple circle denotes a change in the state of the existing monitor; the instrument around the function call corresponds to an instrumentation point that is not minimal in an instrumentation set, so it may only update existing monitors.

\item The final blue circle denotes the second call to the function \lstinline{close_connection}.

\item Finally, enough information has been received for the single monitor involved in verifying this property to reach a verdict (the monitor is collapsed to a single truth value).
\end{itemize}

\subsubsection{Sequence of Database Operations}

\begin{figure}[ht]
\begin{lstlisting}
database = 1
data = [1.1, 1.3, 1.6]
for datum in data:
    operation(datum)
\end{lstlisting}
\caption{\label{fig-code-example-2}Code that, for a given database, iterates through some data and calls a function called \lstinline{operation}.}
\end{figure}

Figure \ref{fig-code-example-2} is a more complex code snippet, and serves as an example of the use of unbounded future time operators in CFTL, in the form of the property:

\begin{equation}\label{eq-code-example-2-property}
\phi \equiv \forall q \in S_d, \forall t \in \text{future}_{\Delta\tau}(q, \text{operation}) : q(\text{database}) = 1 \implies d(t) \in [0, 1].
\end{equation}

\begin{figure}[ht]
\centering
\includegraphics[width=0.8\linewidth]{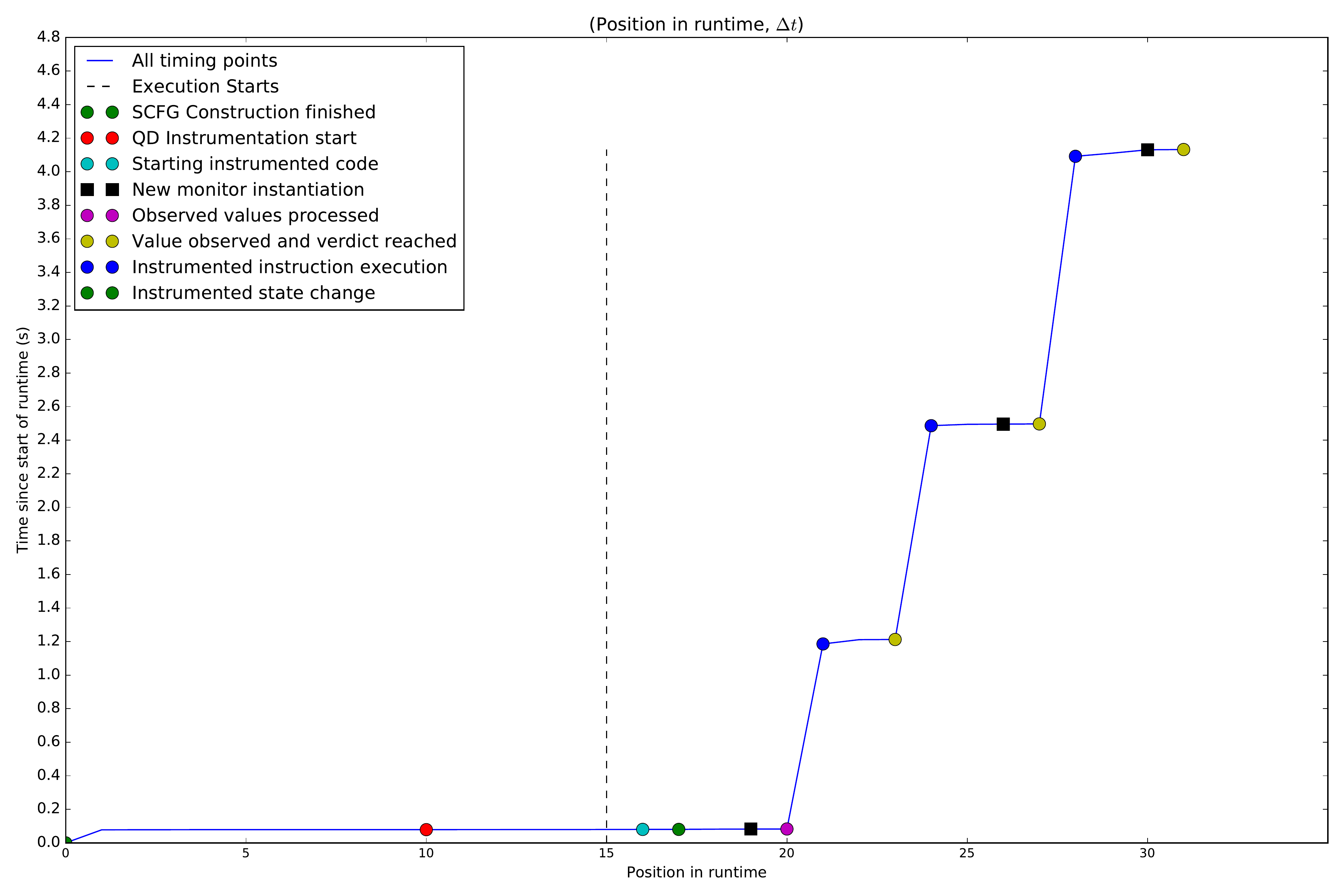}
\caption{\label{fig-code-example-2-plot}Plot showing the progress of verification, with key points in the process followed by the verification tool highlighted.}
\end{figure}

Figure \ref{fig-code-example-2-plot} shows the progress of the verification tool when run on the code in Figure \ref{fig-code-example-2} with the property in Equation \ref{eq-code-example-2-property}.  The meanings of the markings on the plot are all consistent with those in Figure \ref{fig-code-example-1-plot}.

A distinguishing feature of this plot is the greater number of monitors (black square) being instantiated; the property quantifies over both states and transitions in the future, hence every call to \lstinline{operation} in the for-loop in Figure \ref{fig-code-example-2} must contribute to the verification of the property in Equation \ref{eq-code-example-2-property}.  The first black square denotes instantiation of the first monitor when \lstinline{database = 1} is executed.  After this, \lstinline{operation(1.1)} is received by the monitoring thread with a monitor already existing, so the monitor is updated and reaches a verdict.

The second call, \lstinline{operation(1.3)}, is received and corresponds to an instrumentation point that is minimal in the instrumentation set for the bind variable $t$, hence is allowed to trigger instantiation of new monitors.  It does precisely this, using the configuration stored from the old (fully-collapsed) monitor to instantiate a new monitor, hence the second black square after the execution of the second function call.  The third follows in a similar fashion.  Note that the implementation of \lstinline{operation} in this artificial example is simply to delay by the amount of time given in the argument, leading to the increasingly large jumps in Figure \ref{fig-code-example-2-plot}.

\subsubsection{Shared Resource Control}

\begin{figure}[ht]
\begin{lstlisting}
authenticated=1
if authenticated:
    existing_locks = query("get locks", [])
    if len(existing_locks) > 0:
    	# do nothing
    	print("LOCKS EXIST")
        pass
    else:
    	print("NO LOCKS EXIST")
        lock = new_lock()
        query("write", [lock])
else:
    # do nothing
    pass
\end{lstlisting}
\caption{\label{fig-code-example-3}Code that controls access to a shared resource.}
\end{figure}

\begin{equation}\label{eq-code-example-3-property}
\phi \equiv \forall q \in S_a, \forall t \in \text{future}_{\Delta\tau}(q, \text{query}) : q(\text{authenticated}) = 1 \implies d(t) \in [0, 1].
\end{equation}

The property in Equation \ref{eq-code-example-3-property} is immediately similar to that in Equation \ref{eq-code-example-2-property} in that it requires quantification over unbounded future time.

\begin{figure}[ht]
\centering
\includegraphics[width=0.8\linewidth]{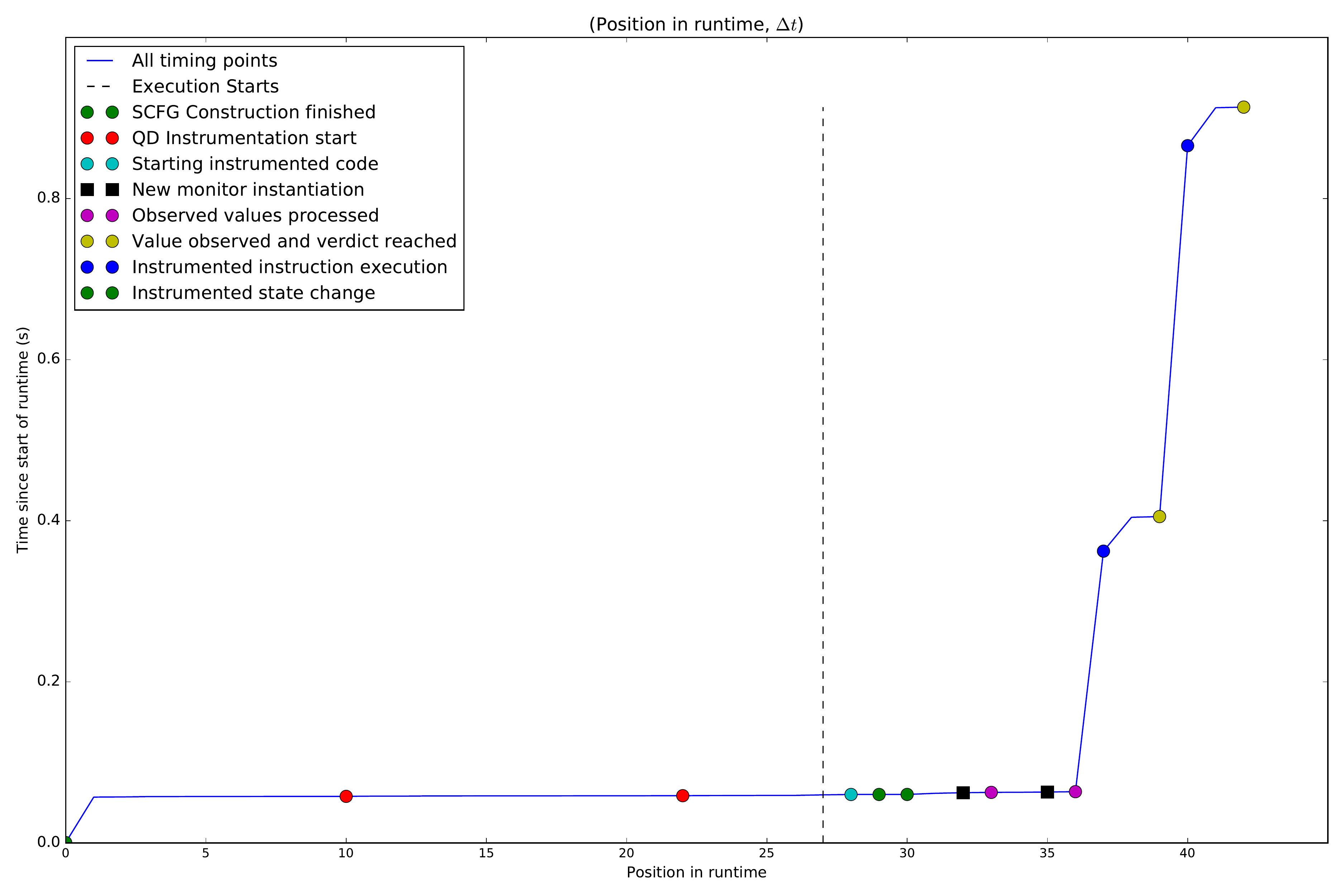}
\caption{\label{fig-code-example-3-plot}Plot showing the progress of verification, with key points in the process followed by the verification tool highlighted.}
\end{figure}

The difference here is that more pressure is placed on the SCFG computed for this code; luckily the presence of the \lstinline{pass} statement in the body of the else-clause on the top level prevents branching from being modelled in the SCFG.  For example, if the second branch does nothing, then it will never be explored and so there is no point in modelling this in the SCFG.

\chapter{Context}\label{context}

This chapter describes the positioning of this work in the existing research.  The chapter opens in Section \ref{section-new-logic} by examining the differences between the new logic introduced for description of state and time constraints, and existing timed temporal logics.

Both the semantics of the new logic and the method of instrumentation given in Chapter \ref{chapter-instrumentation} rely on a variant of the traditional Control Flow Graph, and so the chapter is continued by Section \ref{section-abstract-cfg} describing where this fits into literature from Symbolic Execution and Runtime Verification using Static Analysis.

The chapter concludes by briefly discussing the differences between the monitoring mechanism for CFTL and monitors required for other temporal logics.

\section{Control Flow Temporal Logic}\label{section-new-logic}

CFTL is designed specifically to describe properties regarding constraints over state changes and (mostly) function calls.  To this end, a run of a function being verified according to a property is modelled by a sequence of states with transitions between them.  States are \textit{instantaneous} descriptions of the data in a program; transitions are the computation that happens to move from one state to another, hence have \textit{duration} and are not \textit{instantaneous}.

In this state and transition-based model, function calls become types of transitions.  In the semantics of CFTL, it makes sense to define constraints over transitions in general, hence constraints on function calls, despite being the most common use, are not the only constraints that can be placed.

Other logics that deal with time constraints, such as Metric Temporal Logic \cite{Thati2004} (Section \ref{subsection-mtl}), Timed Linear Temporal Logic \cite{Bauer} (Section \ref{subsection-tltl}) and Real Time Logic \cite{Jahanian} (Section \ref{subsection-rtl}) are defined at a level of abstraction that means using them to express the properties expressible using the work in this report would require an additional layer to serve as an adapter.

Additionally, CFTL is coupled tightly with the control flow of the code over which a formula is verified.  For example, its future-time operators (an example being $\text{next}_{\Delta\tau}$ in Equation \ref{eq-example-property}) are defined in terms of a graph extracted statically from the code under scrutiny.  Given this tight coupling with the control flow of a program, built into CFTL is the ability to not just describe relationships between atoms (eg, $p \implies q$ with no regard for the things to which $p$ and $q$ correspond), but also to describe from where they are derived.  In contrast, the existing logics examined in this chapter disregard the problem of deciding from where in a program/when in its runtime information can be taken.

An example of the combination of a state/transition model with tight coupling with the control flow is now presented.  Consider the property written in English:

\begin{displayquote}
	\textit{Every time the value to which a variable $a$ is mapped is changed, if the new value is in $K$, then the next call to the function $f$ should take an amount of time in $J$.}
\end{displayquote}

This can be written in CFTL as

\begin{equation}\label{eq-example-property}
\forall s \in S_a : s(a) \in K \implies d(\text{next}_{\Delta\tau}(s, f)) \in J.
\end{equation}

This is read as ``for every state $s$ that changes $a$, $s(a) \in K$ in that state implies that the next call to $f$ after the state $s$ is attained must take an amount of time in $J$''.

Now, in terms of expressing the same property in an existing temporal logic, this is more difficult.  One has to introduce a layer of abstraction and, sometimes, deal with lack of semantics based on transitions.  Additionally, traces are restricted to a single scope; if nested calls are captured in a new scope, the property being checked ceases to be regular and cannot be expressed in some of the logics in this chapter.  The implication of this is that, if a state describes a function call beginning, the next state is necessarily the function call ending.

\subsection{Metric Temporal Logic}\label{subsection-mtl}

Metric Temporal Logic\cite{Thati2004, Koymans1990} has a semantics defined in terms of a finite timed sequence of states $(\pi, \tau)$ where: $\pi = (\pi_1, \pi_2, \dots, \pi_n)$ where each $\pi_i \in 2^{AP}$ for a finite set of propositional atoms $AP$; and $\tau \in \mathbb{R}^n$.  Hence, \textit{states} in the context of MTL are sets of propositional atoms with associated timestamps.  Given the semantics in \cite{Thati2004} and the sequence of states $\pi$, one might choose to define the property in Equation \ref{eq-example-property} using MTL by

\begin{equation}\label{eq-example-tltl}
G (x_K \implies F_\infty (f_{\text{start}} \land X_J(f_{\text{end}}))).
\end{equation}

Informally (see \cite{Thati2004} for formal semantics of MTL operators), $G \phi$ asserts that $\phi$ always holds; $F_\infty \phi$ asserts that  $\phi$ \textit{eventually}\footnote{$F_I$ is represented in MTL by $\top U_I \phi_2$, which is read as ``$\top$ holds until $\phi_2$ holds, within the interval $I$''.} (unboundedly) holds; and $X_J \phi$ asserts that, in the next state, $\phi$ holds such that that the time difference between the next state and the current one is within $J$.

Extra atoms have been added (this is the \textit{adapter} discussed earlier): $x_K$ is true iff $x$ has changed and, as a result, $x \in K$ holds; $f_{\text{start}}$ is true iff a state is the latest before a call to $f$; and $f_{\text{end}}$  is true iff $f$ was just called.

Equation \ref{eq-example-tltl} and the surrounding work necessary to allow the expression of the desired property makes it clear that CFTL makes defining certain state and time constraints on program runs more straightforward.

\subsection{Timed Linear-Time Temporal Logic}\label{subsection-tltl}

Timed Linear-Time Temporal Logic \cite{Bauer} is a linear time logic defined on infinite \textit{timed words} (sequences $\sigma$ of $\sigma_i \in \Sigma \times \mathbb{R}_{\ge}$, hence $\sigma_i = (a_i, t_i)$ for some event $a_i$ in the finite alphabet $\Sigma$ and some positive real-numbered timestamp $t_i$).  The set of infinite timed words is often denoted by $T\Sigma^\omega$, where $\Sigma^\omega$ is the set of infinite words over the alphabet $\Sigma$.

Based on the timed word-based semantics, TLTL can be used to describe constraints such as ``$q$ should be true within 5 units of time of $p$ being true'' using its definition of clocks; the $\vartriangleleft_a$ and $\vartriangleright_a$ operators give the times since the previous and until the next occurrences of $a$, respectively.

Further, TLTL can mirror the quantification over states seen in Equation \ref{eq-example-property} by means of implication

$$\square(P \implies Q)$$

which can be read as ``everytime $P$ holds, $Q$ must also hold''.  TLTL does not, however, have a way to explicitly talk about transitions between its states.  Consider the eventuality that one wishes to place a time constraint over a function call:

The only way to talk about this in TLTL would be to force creation of states at either side of the transition and place a property between those.  For example,

$$\square(f_{\text{start}} \land \vartriangleright_{f_{\text{end}}} \in J).$$

Here, the invention of two atoms $f_{\text{start}}$ and $f_{\text{end}} \in \Sigma$ was again required in order for states $\sigma_i, \sigma_{i+1}$ to exist with $\sigma_i = f_{\text{start}}$ and $\sigma_{i+1} = f_{\text{end}}$.  The duration of the call of $f$ would then be modelled by $t_{i+1} - t_i$.  This could be extended to express the property in Equations \ref{eq-example-property} and \ref{eq-example-tltl} by writing

\begin{equation}\label{eq-example-mtl}
\square(x_K \implies F(f_{\text{start}} \land \vartriangleright_{f_{\text{end}}} \in J)))
\end{equation}

where $x_K$ is the same as in Equation \ref{eq-example-tltl}.  This can be read as ``whenever $x \in K$, eventually (unboundedly) we have a call to $f$ and then, within a length of time contained by $J$, $f$ returns''.

In the case of TLTL, it would have to be guaranteed by instrumentation that $f_\text{end}$ was in fact the matching return of the call that triggered the atom $f_\text{start}$ to be true.  Since nested calls are not considered, if $f_\text{start}$ holds in a state, the next state necessarily sees $f_\text{end}$ hold, hence $\vartriangleright_{f_\text{end}}$ gives the time until the next state is attained.  In this case, this is the function call duration.

It should be noted that TLTL does not have bounded versions of the standard temporal operators from untimed LTL, hence $F$ in this case is the same operator as $F_\infty$ in MTL.  After this is done, it is again clear that expressing certain time constraint properties in TLTL is not as intuitive as in the new logic presented.

\subsection{Real Time Logic}\label{subsection-rtl}

Real Time Logic \cite{Jahanian} is a first order logic with a timing element that takes a different approach than those of Timed Linear-Time Temporal Logic and Metric Temporal Logic.  Notably, its semantics are very different; they are based on event occurrences $(e, i, t)$ for an event $e$, the occurrence index $i$ and a timestamp $t$, where a \textit{computation} is a sequence of such event occurrences $\sigma = (\sigma_0, \sigma_1, \dots)$.  Defined on such a sequence are three rules, which assert that 1) time is non-decreasing, 2) at most one value of time may be associated with each event occurrence and 3) every occurrence with an index $i > 1$ has a preceding occurrence.

The distinguishing feature of RTL is the \textit{occurrence function}, $@ : E \times \mathbb{Z}^+ \to \mathbb{N}$, which is a map from pairs consisting of an event $e$ and an occurrence index $i$, to the natural-numbered time at which the occurrence is observed.  This is a key difference from the logic presented in this report; in RTL, one can explicitly obtain the time at which something happens.

However, both the way in which formulas in this logic express time constraints is less intuitive than in CFTL, and there is no way to discuss state.  This means the property in Equation \ref{eq-example-property} cannot be precisely written in RTL, for Equation \ref{eq-example-property} places no time constraint on the gap between the call to $f$ and the state change resulting in $x \in K$.

\subsection{Metric Dynamic Logic and Linear Dynamic Logic}\label{subsection-mdl-ldl}

Linear Dynamic Logic \cite{DeGiacomo} merges regular expressions with Linear Temporal Logic; one can place LTL formulas inside regular expressions and the LDL semantics turns satisfaction of the formula into matching of the sequence of states (sets of atomic propositions, as in the standard LTL semantics given in \cite{Bauer}) with the regular expression.  For example, the Kleene star can be used in place of the \textit{always} temporal operator: $\phi^*$ in LDL is equivalent to $\square(\phi)$ in LTL.

As an example, consider a property ``$\phi$ should hold in the future but, until it does, the states should match the pattern $\rho$''.  This can be written in LDL as follows:

\begin{equation}\label{eq-example-ldl}
\langle \rho \rangle \phi.
\end{equation}

LDL, however, has no notion of time; it simply mixes regular expressions with LTL and so the semantics are defined over sequences of states.  This way, ordering is of course preserved but there is no element of time.

Metric Dynamic Logic \cite{Basin} extends LDL to have a notion of time in a similar way to how Metric Temporal Logic extends Linear-Time Temporal Logic; an interval is placed on some of the operators: $\langle \rho \rangle_I \phi$ expresses the same property as Equation \ref{eq-example-ldl}, but with the additional constraint that $\phi$ should hold within the time interval $I$.  One might then express the property from Equation \ref{eq-example-property} as:

\begin{equation}\label{eq-example-mdl}
x_K \implies ( \langle \text{true}^* \rangle_\infty ( \langle f_{\text{start}} \rangle_J ( f_{\text{end}} ) ) )
\end{equation}

which can be read as ``if $x_K$ holds, then the state sequence from then on should match an arbitrary gap, an occurrence of $f_{\text{start}}$, and, finally, an occurrence of $f_{\text{end}}$ in time bounded by $J$'' with the slight abuse of notation: writing $\langle\rangle_\infty$ to mean $\langle \rangle_{[0, \infty)}$.  In the case of MDL, \textit{occurrence} means that the atomic proposition is true in a given state (the same semantics used by untimed LTL in \cite{Bauer}).  Again, applying regular expressions results in a logic that is not as intuitive as the logic presented for the types of properties of interest in this report.

\subsection{CARET and Recursive Automata}\label{subsection-caret}

CARET \cite{Alur} is a logic that has no timing element, but deals with the inherent inability to deal with nested function calls present in the other logics in this chapter.  It semantics differs from the other logics discussed so far, and the semantics of CFTL, in that a \textit{structured computation} is used as the structure, rather than the typical sequence of states/times.  \cite{Alur} introduces an augmented alphabet, that is, an alphabet $\Gamma$ augmented by product with $\{\text{call}, \text{ret}, \text{int}\}$.  Symbols augmented with \textit{call} denote a change of scope by invocation of a new function, those augmented with \textit{ret} denote a return to the previous scope and those augmented with \textit{int} represent instructions that perform no change on the scope.

This is integrated with the notion of recursive automata; automata for which certain states trigger \textit{instantiation} of other automata.  This is how function calls are modelled, and how the context-free properties resulting from talking about nested function calls are dealt with.

One shortcoming of CARET is its inability to model time.  For example, one can talk about nested function calls, but one cannot place time constraints on them.  Further, CARET is restricted to function calls, whereas CFTL works with \textit{transitions}, which can be any kind of computation that happens in a program to move from one state to another.

\subsection{Freeze Quantification and Dynamic State Variables}\label{subsection-fq-dsv}

Freeze Quantification \cite{Alur1994, Henzinger}, informally, is a mechanism for binding to a variable the time at which a formula begins to hold.  It is the distinguishing feature of the Timed Propositional Temporal Logic introduced in \cite{Henzinger}.  In \cite{Henzinger}, a variable bound by Freeze Quantification is placed after a temporal operator to mean that this variable should take the value of the time at which the truth value of the subformula is being checked.  For example, one could express the property that ``$\phi$ holding should result in $\psi$ holding no more than 5 seconds later'' by

\begin{equation}\label{eq-example-freeze-quantification}
\square x. (\phi \implies \lozenge y.( \psi \land y - x \le 5)).
\end{equation}

This formula can be broken down as such:

\begin{itemize}
	\item $x$ is bound to the current time in which the subformula $(\phi \implies \lozenge y.( \psi \land y - x \le 5))$ is considered.
	\item It is then tested whether, when $\psi$ holds, it is also the case that $y - x \le 5$ \textit{at the moment that $\psi$ begins to hold}.
\end{itemize}

Alternatively, \cite{Alur1992} (a review paper, which also discusses Freeze Quantification) discusses the notion of Dynamic State Variables, also known as Clock Variables, with the notation being called \textit{explicit-clock} notation (based on the fact that it allows the explicit reference to time).

Clock variables work by quantifying formulas over a time domain, and binding values from the time domain to clock variables.  For example, taken from \cite{Alur1992}, one might express the property ``if $p$ holds, then $q$ should hold no more than 3 seconds later'' as such:

\begin{equation}\label{eq-example-clock-variables}
\forall x : \square ((p \land T = x) \implies (q \land T \le x + 3))
\end{equation}

where $T$ is tested for equality to $x$ (the current binding from the quantification) and, if $p$ is observed at time $x$, the time at which $q$ holds must be no more than 3 units of time greater than $x$.

CFTL uses neither Freeze Quantification, nor Dynamic State Variables; it has been found that neither are needed to describe the properties required.  For, consider again the property in Equation \ref{eq-example-property}.  Then one might express it with Freeze Quantifiers in the Timed Propositional Temporal Logic from \cite{Alur1994, Henzinger} as such:

\begin{equation}\label{eq-example-with-fq}
\square(x_K \implies \lozenge.t(f_{\text{start}} \land \circ.t'(f_{\text{end}} \land t' - t \in J)))
\end{equation}

with a slight abuse of notation to express the duration of the call to $f$ being in the interval $J$, and the same introduction of the atoms $x_K, f_{\text{start}}$ and $f_{\text{end}}$ as before.

Using clock variables, the property could be as such:

\begin{equation}\label{eq-example-with-cv}
\forall t : \square(x_K \implies \lozenge(f_{\text{start}} \land T = t \land \circ(f_{\text{end}} \land T - t \in J)))
\end{equation}

From Equations \ref{eq-example-with-fq} and \ref{eq-example-with-cv}, it is clear that expressing some time constraint properties is not intuitive, and one must again add a layer of abstraction to serve as an adapter to what CFTL otherwise deals with naturally.

\section{An Abstract, Static Model of Programs}\label{section-abstract-cfg}

In order to define both the semantics of CFTL, and to instrument for it, a structure is used that is inspired by the symbolic execution tree \cite{Baldoni, Kinga} found in the symbolic execution literature.  Symbolic execution was first described in \cite{Kinga} as a program verification approach; the basic idea is to input \textit{symbolic terms} into a program, allow the execution tree to expand (as different paths are taken, since they are all explored simultaneously) and perform satisfiability checking on the conditions seen so far on each path.  The work on verification by symbolic execution is not entirely useful to the work in this report, but a variant of the structure central to the theory (the symbolic execution tree) is.

Both \cite{Kinga} and \cite{Baldoni} describe the symbolic execution tree as having vertices for each statement in a program (and the resulting state after the statement is executed).  Edges correspond to control flow; if branching occurs, conditions are attached to the edges and, if no branching occurs (ie, in a basic block), edges just represent the ordering of statements.

The structure in Section \ref{section-static-model} considers a type of control flow graph augmented with a similar structure to the symbolic execution tree.  Vertices correspond to \textit{symbolic states} that are induced by statement execution and edges correspond to the statements that generate the states.  Structure taken from the symbolic execution tree includes conditions being placed on edges; when branching occurs between two symbolic states, the conditions are placed on the edges.

Additionally, since this work considers a mix between a control flow graph and a symbolic execution tree, the global graph structure is closest to the control flow graph in that convergence is possible.  In particular, the symbolic execution tree represents paths through control flow as distinct paths through the tree; this is not the case in this work.  For example, branching in a symbolic execution tree would result in two distinct paths with no possibility for convergence in the future; the structure considered in this report is more a representation of the code, rather than the executions it can generate.

Finally, this variant of the symbolic execution tree and control flow graph is statically computable, and so is used as a basis for the model of program runs given in Section \ref{section-program-run-model}, and is also used in the instrumentation method presented.

\subsection{Influence on Semantics}\label{subsection-influence-on-semantics}

As described in Section \ref{section-new-logic}, existing logics use semantics usually based on sequences of either 1) sets of propositions that are true in the instant to which the state corresponds \cite{Bauer, Thati2004, Koymans1990, DeGiacomo, Basin} (LTL, MTL, LDL and MDL) or 2) symbols from some alphabet of \textit{events} \cite{Bauer} (TLTL).  The only logic seen to share a notion of transitions similar to those considered in this report is RTL \cite{Jahanian} (Section \ref{subsection-rtl}).  Specifically, its notion of \textit{actions}, which are defined informally as the computation taking place between two \textit{events}, draws similarities with the transitions in the semantics defined later.

In a new approach, the preceding chapters discuss the combination of an alphabet of \textit{critical symbols} with the augmented static graph representation of the program under scrutiny.  This forms a state and transition-based model of a run of that program, over which the truth value of a formula can be defined.  This notion of critical symbols is also used in the static analysis-based approach to instrumentation discussed in Section \ref{subsection-instrumentation-planning}.

\subsection{Instrumentation Planning}\label{subsection-instrumentation-planning}

The work on design of logics examined in Section \ref{section-new-logic} does not typically address how instrumentation is performed; it starts by defining the grammar for the logic, gives the semantics, and sometimes finishes by giving a monitoring algorithm.

The basis of the instrumentation part of this work is the symbolic control flow graph (Section \ref{section-static-model}), and draws similarities with work done on determining a \textit{Minimal Sampling Period} (MSP)\cite{Bonakdarpour}.  That work attempts to use a static graph representation of a program (similar to that in this report) to optimally instrument for formulas in LTL.  This is a distinguishing feature of the work presented here: the instrumentation is tightly coupled with the logic and its semantics.

In particular, \cite{Bonakdarpour} presents algorithms (which must ultimately use heuristics, since the optimisation problem derived is NP-Complete) that attempt to determine the optimal amount of time to wait to take a sample of a variable $x$ from some critical set $U_\phi$ (ultimately the alphabet over which the LTL formulas they consider are written).  The graph they derive statically is a control flow graph on which some division is performed to be sure that vertices containing instructions that modify some $x \in U_\phi$ contain \textit{only} those instructions.  These vertices are then called \textit{critical vertices} and are used in the determination of the MSP.

A major source of contrast between \cite{Bonakdarpour} and the work given here is the inclusion of reachability analysis in the instrumentation strategy presented in Chapter \ref{chapter-instrumentation}.  In particular, the future-time operators (such as $\text{next}_{\Delta\tau}$ in Equation \ref{eq-example-property}) require traversal of the augmented control flow graph to decide where to instrument.  From this, some interesting problems arise when one considers the lack of a total order of symbolic states/transitions due to no runtime information being available.

\section{Monitor Synthesis}\label{section-monitor-synthesis}

An overview of constructing an optimised (in terms of size) monitoring mechanism is given in \cite{Wolper}.  The focus is on untimed Linear Temporal Logic (the verification of which being less involved than its timed counterpart TLTL) with infinite words, and takes the standard approach: a B\"{u}chi automaton is constructed (B\"{u}chi automata are defined over infinite words\footnote{Their acceptance condition allows \textit{cycling} of input symbols around the automaton, hence a finite state automaton can accept an infinite word.}) and then minimised (by considering isomorphism in parts of the automaton's structure).

In general, \cite{Wolper} uses the closure of a formula (the set of all subformulas, analogous to the power set) in construction of the automaton.  Specifically, letting $\pi$ be the sequence of states discussed in Section \ref{section-new-logic}, they define a \textit{closure labelling} $\sigma : \mathbb{N} \to 2^{\text{cl}(\psi)}$ where $\text{cl}(\psi)$ is the closure of $\psi$.  Considering this alongside the rule that, if $\phi \in \sigma(i)$, then $\pi(i) \models \phi$ where $\pi(i) = (\pi_i, \pi_{i+1}, \dots)$, they define a sequence of lemmata that ultimately allow relatively straightforward construction of a B\"{u}chi automaton for $\psi$.

As an example, consider $\psi = \square(\phi)$.  Then, if $\phi \in \sigma(j)$ for some $j$, $\phi \in \sigma(k)$ for every $k \ge j$ also.  This follows the rule described above: $\phi \models \tau(k)$ for every $k \ge j$ must be the case if $\phi \in \pi(j)$.

CFTL uses temporal operators, but the current instrumentation strategy removes the need for construction of complex monitors.  Because of this, the current strategy is simply a form of formula rewriting, which is accomplished by using formula trees and progressively collapsing them as data to resolve the atoms in the formula is observed.  The nature of the logic means that formula rewriting does not result in explosion of formula size, and the optimisation using the formula closure means formula collapse (which must happen at runtime) speeds up collapse considerably.

\newpage

\bibliography{references}{}
\bibliographystyle{unsrt}

\end{document}